\documentclass{article}
\usepackage[utf8]{inputenc}
\usepackage[numbers]{natbib}
\usepackage{graphicx}
\usepackage{geometry}
\usepackage[font=scriptsize]{caption}
 \usepackage{float}
 
\usepackage{titlesec}

\usepackage{pdfpages}

\titleformat{\paragraph}
{\normalfont\normalsize\bfseries}{\theparagraph}{1em}{}
\titlespacing*{\paragraph}
{0pt}{3.25ex plus 1ex minus .2ex}{1.5ex plus .2ex}
 
\usepackage{caption}
\usepackage{subcaption}
 
\usepackage{amsmath,amssymb}
\usepackage{bbm} 
\usepackage{amsthm} 
\numberwithin{equation}{section}
\usepackage{mathrsfs} 
\allowdisplaybreaks
\usepackage{mathtools}
\usepackage[ruled,vlined]{algorithm2e}

\SetKwInput{KwInput}{Input}                
\SetKwInput{KwOutput}{Output}              

\newcommand{\deriv}[1]{\frac{d}{d#1}}

\newcommand\given[1][]{\:#1\vert\:}

\DeclareMathOperator{\E}{\mathbb{E}}
\DeclareMathOperator{\Prob}{\mathbb{P}}
\DeclareMathOperator*{\argmax}{arg\,max}

\newcommand\norm[1]{\left\lVert#1\right\rVert}

\newtheorem{definition}{Definition}[section]
\newtheorem{proposition}{Proposition}[section]

\usepackage{tcolorbox}

\usepackage{bm}

\usepackage{environ}
\NewEnviron{eqs*}{%
\begin{equation*}\begin{split}
  \BODY
\end{split}\end{equation*}
}

\usepackage{titling}

\usepackage{authblk}

\usepackage[breaklinks,plainpages=false,hyperfootnotes=false]{hyperref}
\hypersetup{
  colorlinks   = true, 
  urlcolor     = blue, 
  linkcolor    = red, 
  citecolor   = blue 
}

\author[1]{R.L. Gudmundarson \footnote{E-mail address:
    \href{mailto:rlg2000@hw.ac.uk}{rlg2000@hw.ac.uk}}}
\author[2]{M. Guerra \footnote{E-mail address: \href{mailto:mguerra@iseg.ulisboa.pt}{mguerra@iseg.ulisboa.pt}}}
\author[2]{A. B. de Moura \footnote{E-mail address:
    \href{mailto:amoura@iseg.ulisboa.pt}{amoura@iseg.ulisboa.pt}}}
    \affil[1]{Edinburgh Business School, Heriot-Watt University}
\affil[2]{ISEG-School of Economics and Management, Universidade de Lisboa; REM - 
Research in Economics and Mathematics, CEMAPRE}

\title{Minimizing Ruin Probability Under Dependencies for
Insurance Pricing}
\date{\vspace{-6ex}}

\begin{document}
\maketitle
\allowdisplaybreaks

\addcontentsline{toc}{section}{abstract}
\begin{abstract}
In this work the ruin probability of the Lundberg risk process is used as a criterion for determining the optimal security loading of premia in the presence of price-sensitive demand for insurance. Both single and aggregated claim processes are considered and the independent and the dependent cases are analyzed. For the single-risk case, we show that the optimal loading does not depend on the initial reserve. In the multiple risk case we account for arbitrary dependency structures between different risks and for dependencies between the probabilities of a client acquiring policies for different risks. In this case, the optimal loadings depend on the initial reserve. In all cases the loadings minimizing the ruin probability do not coincide with the loadings maximizing the expected profit
\end{abstract}


\pagenumbering{arabic} 

\section{Introduction}
\label{section:litReview}

Insurance is based on the idea that society asks for protection against unforeseeable events which may cause serious financial damage. Insurance companies offer a financial protection against these events. The general idea is to build a community where everybody contributes a certain amount and those who are exposed to the damage receive financial reimbursement \cite{wuthrich2019non} . 
\smallskip

When (non-life) insurers set premium prices they usually start by finding the so-called pure premium, which is the expected value of the total claims that will occur in one time unit. However, when pricing insurance policies, insurers must take into account the risk associated with the policy as well as additional costs (e.g. operational cost, capital cost, etc.). Therefore, a so-called security loading is added to cover the risk and additional costs. The security loading is often calculated using some premium calculation principle, 
and the insurance premium is obtained once the security loading has been determined and added to the pure premium. The main concerns are usually whether the loading is an appropriate measure of the risk and which premium principle to choose. The higher the loading the higher the premium and consequently, the underwriting risk will be lower. However, if the premium price is too high then the exposure will be too low due to competition, and the operational cost of the insurer will engulf the premium income resulting in financial instability. Therefore, insurers usually require sophisticated premium calculations in order to secure stability.
\smallskip

Collective risk models are fundamental in actuarial science to model the aggregate claim amount of a line of business in an insurance company. The collective risk model has two random components, the number of claims and the severity of claims, and is usually modelled with a compound process \cite[Chapter~3]{kaas2008modern}. The classical Lundberg risk process has been studied extensively and there exist many variations, for example including reinsurance or investments \cite{hipp2000optimal}. It assumes that premia come in a continuous stream while claims happen at discrete times according to a Poisson distribution.
\smallskip

Another common assumption is that the risk can be divided into groups of homogeneous risks such that the pure premia and security loadings can be estimated separately for each risk group. The pure premia of these individual groups are usually modelled with generalized linear models (GLM). GLM's have been applied extensively in actuarial work and a good overview is provided in \cite{ohlsson2010glm, yao2013generalized}. Traditional risk theory has usually assumed independence between risks due to its convenience, but it is generally not very realistic. Claims in an insurer's risk portfolio are correlated as they are subject to the same event causes \cite{campana2005distortion}. Completely homogeneous risk groups are extremely rare and dependence among risks has become a flourishing topic in actuarial literature \cite{Copulas_Barning}. Dependence has mostly been measured through linear correlation coefficients \cite{britt2005linear}. The popularity of linear coefficient is mainly due to the ease with which dependencies can be parameterized, in terms of correlation matrices. Most random variables, however, are not jointly elliptically distributed and it could be very misleading to use linear coefficients \cite{straumann2001correlation}. This motivated the use of concordance measures. Two random variables are concordant when large values of one go with large values of the other \cite{nelsen2007introductionCopulas}. The Lundberg risk model is a Lévy jump process, \cite{ContTankov2004Jumps} which means that the dependency of two claim processes is best explained through their Lévy measure \cite{tankov2016levy}. This study will not go into details about Lévy processes, but both \cite{ContTankov2004Jumps} and \cite{papapantoleon2008introduction} provide a very good introduction, and \cite{Copulas_Barning, BAUERLE2011398, sato1999levy, van2012parameter, avanzi_cassar_wong_2011} are examples of applications of Lévy copulas to risk processes. For example, van Velesen \cite{van2012parameter} showed how Lévy copulas can be used in operational modelling and discussed how dependence is implied by the Lévy copula. In this work we consider bivariate claim processes, but the presented theory can be straightforwardly extended to multiple claim processes.
\smallskip

Ruin probability is a classical measure of risk and has been extensively studied \cite{kaas2008modern, hipp2000optimal,kasumo2018minimizing, trufin2011properties}. Although there is no absolute meaning to the probability of ruin, it still measures the stability of insurance companies. A high ruin probability indicates instability, and risk mitigation techniques should be used, like reinsurance or raising premia \cite{kaas2008modern}. Most non-life insurance products have a term of one year and therefore it can be argued that the one year ruin probability should be used. The one year ruin probability is the probability that the capital of an insurance company will hit zero within one year. However, the appropriateness of risk measures defined over fixed time horizons can be questioned, since ruin in a given time span can be minimized by increasing the probability of ruin in the aftermath of that period. Lundberg concluded that the actual assumptions behind the classical collective risk model are in fact less restrictive when time-invariant quantities like the infinite time ruin probability are considered \cite{trufin2011properties}. Therefore, we focus on the infinite time ruin probability in this paper.
\smallskip

In this work, the optimal loadings based on two strategies are derived, and compared. One strategy maximizes the profit and the other minimizes the ruin probability. We show that the two loading strategies give different results. Furthermore, we show how the optimal loading with respect to the ruin probability can be found and compare it to the one obtained when the expected profit is maximized. We consider dependencies and illustrate how Lévy copulas can be used to model claim process dependencies and how dependencies can affect the riskiness of the insurance portfolio. We take this idea further and consider dependency between the acquisition of insurance for different risks by policyholders. This is a realistic assumption as policyholders usually buy multiple insurance products from the same insurance company. We also take into account the fact that the market risk process and the company's risk process are not the same, and how the company's risk process depends on its exposure to the market. This is, to our knowledge, the first analysis of the interplay of the ruin probability, the dependency structure of claim, and the dependency structure of acquisition of insurance. We demonstrate that even if there is a strong dependency between insurance products within the market, small insurance companies have less dependency and therefore less risk than bigger insurance companies, provided the dependency between acquisition of insurance for different risks is not too strong. 
\medskip

The paper is organized as follows: Section \ref{section:min_ruin} contains some background material about ruin probabilities in the Lundberg process and aggregation of compound Poisson processes. Section \ref{section:demand} deals with the single-risk case. We characterize the optimal loading and compare it with the loading maximizing the expected profit.  Section \ref{section:companyRisk} handles the multiple risks case. We show how the dependency structure existing in the market (i.e. the general population) translates into the risk exposure of the company through its market shares on different risks and the likelihood that clients acquire insurance for more than one risk. Section \ref{seq:numerical} contains a numerical illustration. A numerical scheme to compute the ruin probabilities is given in the appendix.


\section{Preliminaries}
\label{section:min_ruin}

\subsection{Claim and Surplus Processes}

The Lundberg risk model describes the evolution of the capital of an insurance company and assumes that exposure is constant in time, losses follow a compound Poisson process, and premia arrive at a fixed continuous rate:

$$ X_t = u + ct-\sum_{i=0}^{N_t} Y_i = u + ct - S_t, \qquad Y_0 \coloneqq 0  ,$$

where $u$ is the initial surplus, $c$ is the risk premium rate, $N_t$ is a time homogeneous Poisson process with intensity parameter $\lambda$, and  $Y_i$ are i.i.d. random variables representing the severity of claim $i$, $i = 0,\dots, N_t$. Here it is assumed that $Y_i$ are positive. In the following sections, $Y$ denotes an arbitrary random variable with the same distribution as any $Y_i$. The severity distribution is denoted as $F(x)$ and the severity survival distribution as $\overline{F}(x)$. $S_t$ is a compound Poisson process and thus $X_t$ is a stochastic process (sometimes called the surplus process) representing the insurance wealth at time $t$. $X_t$ increases because of earned premia and decreases when claims occur. When the capital of an insurance company hits zero, the insurance company is said to be ruined. Formally, the ruin probability is defined as follows.

\begin{definition}[Probability of Ruin]
\label{SimpleRuin}
Let $(\Omega, \mathcal{F}, \{\mathcal{F}_t\}_{t\geq0}, \mathbbm{P})$ be a filtered probability space and $X = (X_t)_{t \in [0, \infty[}$ a surplus process which is adapted and Markov with respect to the filtration. The state space is $(\mathbbm{R}, \mathcal{B}(\mathbbm{R}))$. If X is time homogeneous, the infinite time ruin probability is the function $V: \mathbbm{R} \mapsto [0,1]$ such that

$$V(x) = \Prob \big( \exists s\in[0,+\infty[:X_s \leq 0  \given[\big]   X_0 = x \big), \quad x \in \mathbbm{R} .$$

\end{definition}

Sometimes it is useful to use the survival (non-ruin) probability, defined as $\overline{V}(x) = 1-V(x)$. The ruin probability can be calculated using the following integro-differential equation \cite{grandell2012aspects}.

\begin{proposition}
\label{prop:Ruin_lundberg_inf}
Assume that $X_t$ is defined as above and the premium rate satisfies $c > \lambda \E[Y]$. If $V \in C^1(]0, \infty[),$ then the probability of ruin with infinite time horizon satisfies the following equation:
\begin{align}
\label{eq:Ruin_lundberg_inf}
    0 =  c \deriv{x}V(x) + \lambda \bigg( \int_{0}^{x} V(x-y)dF(y) - V(x) + 1-F(x) \bigg), \quad x>0,
\end{align}
with the following boundary condition:
\[
\begin{cases}
   V(x) = 1  & x\leq 0,\\
  \lim_{x \to 0^+}V(x) = \frac{\lambda}{c}\E[Y] . \\
\end{cases}
\]
Furthermore, the probability of non-ruin satisfies the following equation:
\begin{equation}
\label{eq:Ruin_lundberg_inf_survival}
    \overline{V}(x)-\overline{V}(\epsilon) = \frac{\lambda}{c}\int_\epsilon^x\overline{V}(x-y)\overline{F}(y)dy
\end{equation}
for $0 < \epsilon \leq x <  + \infty$  with the following boundary condition:
\[
\begin{cases}
   \overline{V}(x) = 0  & x\leq 0,\\
  \lim_{x \to 0^+} \overline{V}(x) =1- \frac{\lambda}{c}\E[Y] . \\
\end{cases}
\]
\end{proposition}
\medskip

A numerical scheme solving equation \eqref{eq:Ruin_lundberg_inf} can be found in Appendix \ref{appendix:NA_ez_Poi_Exp_inf}.

\subsection{Accounting for Claim Dependencies}

Consider the surplus process $\bm{X} = (X_t^{(1)},...,X_t^{(n)}) $ where
\begin{equation}
\label{multiclaims}
    \begin{split}
        X_t^{(1)} &= u^{(1)}+ c^{(1)}t - \sum_{i = 0}^{N_t^{(1)}}Y_i^{(1)} \\
        \vdots &\\
        X_t^{(n)} &= u^{(n)}+ c^{(n)}t - \sum_{i = 0}^{N_t^{(n)}}Y_i^{(n)}  \\
    \end{split}
\end{equation}

If these processes are independent, it is relatively easy to combine them into a single process using the aggregation property of compound Poisson processes as described in W\"{u}trich \cite{wuthrich2019non}. The aggregation property allows the combination of multiple surplus processes into a single risk process as follows:
\begin{equation*}
    X_t = \sum_{j= 1}^n u^{(j)} +\sum_{j = 1}^n c^{(j)}t -\sum_{i = 0}^{N_t}Y_i,
\end{equation*}
where $N_t$ is a Poisson r.v. with $\lambda = \lambda_{1} + ...+ \lambda_{n}$ and $Y_i$ are i.i.d. random variables, which follow the severity distribution $F(x) = \sum_{j = 1}^n \frac{\lambda_j}{\lambda} F_j(x)$. This aggregation property allows us to use the integro-differential equation \eqref{eq:Ruin_lundberg_inf} to calculate the ruin of multiple surplus processes.
\smallskip

If the risks are not independent, then we can use the fact that compound Poisson processes are characterized by their Lévy measure to decompose the claim process into independent processes to which the aggregation property can be applied.
In particular, for $n = 2$ risks, we obtain the decomposition:
\begin{equation*}
    \begin{split}
          X_t = X_t^{(1)} + X_t^{(2)} = u + ct - S_t^{1\perp} - S_t^{2\perp} -S_t^\parallel.
    \end{split}
\end{equation*}
where $S^{1\perp}$ and $S^{2\perp}$ are compound Poisson processes accounting for events concerning only risk 1 and risk 2, respectively. $S^{\parallel}$ is a compound Poisson process accounting for events concerning both risks simultaneously. Furthermore, $S^{1\perp}$, $S^{2\perp}$ and $S^{\parallel}$ are mutually independent. \smallskip

In this section, we briefly explain how this can be achieved. Further details can be found in \cite{tankov2016levy}. We will use the following definitions:

\begin{definition}
The tail integral of a Lévy measure $\nu$ on $[0, \infty]^2$ is given by a function $U:[0, \infty]^2 \mapsto [0, \infty]$
\begin{equation}
\label{eq:tail_int}
    \begin{split}
        & U(x_1,x_2) = 0 \quad if \quad x_1 = \infty \quad or \quad x_2 = \infty ,\\
        & U(x_1,x_2) = \nu\big([x_1, \infty[ \times [x_2, \infty[ \big) \quad for \quad (x_1,x_2) \in ]0, \infty[^2 , \\
        & U(0,0) = \infty.
    \end{split}
    \end{equation}
\end{definition}

\begin{definition}[Lévy Copula for Processes with Positive Jumps]
A two-dimensional Lévy copula for Lévy processes with positive jumps, or for short, a positive Lévy copula, is a 2-increasing grounded function $\mathcal{C}: [0,\infty]^2 \to [0, \infty]$ with uniform margins, that is, $\mathcal{C}(x,\infty) = \mathcal{C}(\infty,x) = x$.
\end{definition}

Similarly to Sklar's theorem for ordinary copulas \cite{nelsen2007introductionCopulas}, it has been shown that the dependency structure of $(X_t^{(1)},X_t^{(2)})$ can be characterized by a Levy copula $\mathcal{C}$ such that $\mathcal{C}(U_1(x_1), U_2(x_2))$ where $U_1$ and $U_2$ are the marginal tail integrals for $X_t^{(1)}$ and $X_t^{(2)}$. If $U_1$ and $U_2$ are absolutely continuous, this Lévy copula is unique, otherwise it is unique on $Range(U_1)\times Range(U_2)$, the product of ranges of one-dimensional tail integrals, \cite[Theorem~5.4]{ContTankov2004Jumps}
\smallskip

Consider a two dimensional claim process:
\begin{equation}
    \label{eq:2DimClaimProcess}
    S_t = (S_t^{(1)}, S_t^{(2)}) = \sum_{i = 0}^{N_t} (Y_i^{(1)},Y_i^{(2)}),
\end{equation}
where $N_t$ is a Poisson process with intensity $\lambda$ and $Y_i = (Y_i^{(1)},Y_i^{(2)})$, $i \in \mathbbm{N}$ are independent random variables with common joint distribution $F_Y$. The components of $S$, $S^{(1)}$ and $S^{(2)}$, are one-dimensional compound Poisson processes with intensities $\lambda_1$ and $\lambda_2$ and severity distributions $F_{Y^{(1)}}$ and $F_{Y^{(2)}}$, respectively. 
We wish to obtain a decomposition:
\begin{equation}
\label{eq:bracketSum}
     (S_t^{(1)}, S_t^{(2)}) = \sum_{i = 0}^{N_t^{1\perp}}  (Y_i^{(1\perp)}, 0) + \sum_{i = 0}^{N_t^{2\perp}}  (0, Y_i^{(2\perp)}) + \sum_{i = 0}^{N_t^{\parallel}}(Y_i^{(1\parallel)}, Y_i^{(2\parallel)}),
\end{equation}
where $\sum_{i = 0}^{N_t^{1\perp}}  Y_i^{(1\perp)}$,$\sum_{i = 0}^{N_t^{2\perp}} Y_i^{(2\perp)}$ and $\sum_{i = 0}^{N_t^{\parallel}}(Y_i^{(1\parallel)}, Y_i^{(2\parallel)})$ are independent compound Poisson processes with intensities $\lambda_1^\perp$, $\lambda_2^\perp$, $\lambda^\parallel$ and severity distributions $F_{Y^{1\perp}}$,$F_{Y^{2\perp}}$,$F_{Y^{\parallel}}$, respectively. In the above setting, we consider \begin{equation}
\label{eq:AllDistZero}
    F_Y(0,0) = F_{Y^{\parallel}}(0,0) = 0, \quad F_{Y^{(1)}}(0) = F_{Y^{1\perp}}(0) = F_{Y^{(2)}}(0) = F_{Y^{2\perp}}(0) 
\end{equation}

A compound Poisson process, $S$, is a Lévy process with Lévy measure $\nu(dx) = \lambda dF(x)$, with tail integral
\begin{equation*}
    U(x_1,x_2) = 
    \begin{cases}
    \lambda \Prob \big(Y^{(1)} \geq x_1, Y^{(2)} \geq x_2 \big) & \textrm{if } x_1 >0 \textrm{ or } x_2 >0 \\
    +\infty & \textrm{if } x_1 = x_2 = 0 .
    \end{cases}
\end{equation*}
The components $S^{(1)}$ and $S^{(2)}$ are independent if and only if $U(x_1, x_2) = 0$ for every $(x_1, x_2) \in ]0, + \infty[$, i.e., if and only if $\lim_{x_1 \to 0^{+}, x_2 \to 0^{+}} U(x_1, x_2) = 0 $.
\smallskip

The Lévy measure of the processes $S^{(i)}$, $i = 1,2$, have tail integrals
\begin{equation*}
    \begin{split}
        U_1(x_1) &= \lambda_1 \Prob \big( Y^{(1)} \geq x_1 \big) = U(x_1,0) \\
        U_2(x_2) &= \lambda_2 \Prob \big( Y^{(2)} \geq x_2 \big) = U(0,x_2)
    \end{split}
\end{equation*}
Taking equation \eqref{eq:AllDistZero} into account, one obtains 
\begin{equation*}
     \lambda_i = \lim_{x_i \to 0^{+}} U_i(x_i), \quad i = 1,2 
\end{equation*}
\begin{equation*}
     \lambda =\lim_{x_1, x_2 \to 0^{+}} \big(U_1(x_1) + U_2(x_2) - U(x_1,x_2) \big) 
     = \lambda_1 + \lambda_2 - \lambda^\parallel
\end{equation*}
\begin{equation*}
    \lambda^\parallel =\lim_{x_1, x_2 \to 0^{+}} U(x_1,x_2) 
\end{equation*}
\begin{equation*}
     \lambda_i^\perp = \lambda_i - \lambda^\parallel, \quad i = 1,2 .
\end{equation*}
The severity distributions $F_{Y^{1\perp}}$, $F_{Y^{2\perp}}$, and $F_{Y^{\parallel}}$ can be recovered from the tail integrals:
\begin{equation*}
    \Prob \big(Y^{1\perp} \geq x_1 \big) = \frac{1}{\lambda_1^\perp} \lim_{x_2 \to 0^{+}} \Big( U_1(x_1) - U(x_1,x_2) \Big)
\end{equation*}
\begin{equation*}
    \Prob \big(Y^{2\perp} \geq x_2 \big) = \frac{1}{\lambda_2^\perp} \lim_{x_1 \to 0^{+}} \Big( U_2(x_2) - U(x_1,x_2) \Big)
\end{equation*}
\begin{equation*}
    \Prob \big(Y^{1\perp} \geq x_1, Y^{2\perp} \geq x_2 \big) = \frac{1}{\lambda^\parallel}  U(x_1,x_2) .
\end{equation*}

If the dependency between $S^{(1)}$ and $S^{(2)}$ is characterized by a Lévy copula, $\mathcal{C}$, i.e. $U(x_1, x_2) = \mathcal{C}(U_1(x_1), U_2(x_2))$ for $(x_1, x_2) \in [0, +\infty[^2$, then the relations above can be written using the Lévy copula and one-dimensional tail integrals:
\begin{equation*}
\lambda^\parallel = \lim_{u_1 \to \lambda_1^{+}, u_2 \to \lambda_2^{+}} \mathcal{C}(u_1, u_2)
\end{equation*}
\begin{equation*}
    \Prob \big(Y^{1\perp} \geq x_1 \big) = \frac{1}{\lambda_1^\perp} \lim_{u_2 \to \lambda_2^{+}} \Big( U_1(x_1) - \mathcal{C}(U_1(x_1),u_2) \Big)
\end{equation*}
\begin{equation*}
    \Prob \big(Y^{2\perp} \geq x_2 \big) = \frac{1}{\lambda_2^\perp} \lim_{u_1 \to \lambda_1^{+}} \Big( U_2(x_2) - \mathcal{C}(u_1, U_2(x_2)) \Big)
\end{equation*}
\begin{equation*}
    \Prob \big(Y^{1\perp} \geq x_1, Y^{2\perp} \geq x_2 \big) = \frac{1}{\lambda^\parallel}  \mathcal{C}(U_1(x_1), U_2(x_2)) .
\end{equation*}

Using the above methodology, the surplus process 
can be represented as:
\begin{equation*}
        X_t = u + ct - \sum_{i = 0}^{N_t^{1\perp}} Y_i^{1\perp} - \sum_{i = 0}^{N_t^{2\perp}} Y_i^{2\perp} -  \sum_{i = 0}^{N_t^{\parallel}} ( Y_i^{1\parallel} + Y_i ^{2 \parallel} ) =
        u +ct - \sum_{i = 1}^{N_t^{*}} Y_i^{*} ,
\end{equation*}
where $u = u_1 + u_2$, $c = c_1 + c_2$, $N^{*}$ is a Poisson process with intensity $\lambda = \lambda_1^\perp + \lambda_2^\perp + \lambda^\parallel$ and $Y_i^{*}$ are i.i.d. random variables with distribution:
\begin{equation*}
    F^{*} = \frac{\lambda_1^\perp}{\lambda} F_{Y^{1\perp}} + \frac{\lambda_2^\perp}{\lambda} F_{Y^{2\perp}} + \frac{\lambda^\parallel}{\lambda} F_{Y^{1\parallel} + Y^{2\parallel}} ,
\end{equation*}
where
\begin{equation*}
    F_{Y^{1\parallel} + Y^{2\parallel}}(x)  = \int_{x_1 + x_2 \leq x} dF_{Y^\parallel}(x_1,x_2).
\end{equation*}

\section{The Optimal Loading for a Single Risk}
\label{section:demand}

An insurer can control the volume of its business through the premium loading $\theta$. A reasonable assumption is that the higher the loading, the smaller the number of contracts in its portfolio, which means that the claim intensity (or business volume) will decrease. Therefore, both the claim intensity $\E^\theta[N_1]$, and the premium rate $c(\theta)$, will depend on $\theta$. It is reasonable to assume that $\E^\infty[N_1] = 0$, as abnormal premium rates will not attract customers \cite{hipp2004insurancecontrol}. To capture these concepts let $\E^\theta[N_1] = \lambda p(\theta)$. Here $\lambda$ is the average number of claims per unit of time for the whole market, and $p(\theta)$ is the probability that a potential claim is filed as an actual claim to the particular insurer under consideration. In other words, $p(\theta)$ reflects the demand or the market share sensitivity to the loading parameter $\theta$. $p(\theta)$ can be interpreted as a probability that a customer buys an insurance product. For example, we may assume that demand of insurance contracts is described by a logit glm model as in Hardin and Tabari \cite{hardin2017renewal}. Thus, $p(\theta)$, will be :
\begin{equation}
\label{eq:logit}
    p(\theta) = \frac{1}{1+e^{\beta_0 + \beta_1 \theta}},
\end{equation}
where $\beta_0$ and $\beta_1$ are determined from the glm and $\theta$ is the loading parameter. $\beta_1$ will be a positive number so $p\to 0$ when $\theta \to \infty$ and $p\to 1$ when $\theta \to -\infty$. Assuming that the company has some fixed costs, independent of the risk exposure, denoted by $r > 0$, the expression for the net premium income becomes:
\begin{equation*}
    c(\theta) = (1 + \theta)\E^\theta[N_1]\E[Y] -r.
\end{equation*}

The following proposition characterizes the behaviour of the solution of equation \eqref{eq:Ruin_lundberg_inf} with respect to the loading $\theta$.

\begin{proposition}
\label{prop:alpha_proof}
If $V(x,\theta)$ satisfies equation \eqref{eq:Ruin_lundberg_inf} then $V(x,\theta)$ is strictly increasing with respect to the parameter $\alpha= \frac{ \E^\theta [N_1]}{c(\theta)}$.
\end{proposition}

\begin{proof}
It is possible to integrate Equation \eqref{eq:Ruin_lundberg_inf} on the interval $]0, x]$ to obtain:
\begin{equation}
\label{eq:omega_prove_1}
    \begin{split}
        V(x, \theta) = \frac{ \E^\theta [N_1]}{c(\theta)} \Bigg( \E[Y] + \int_0^x \Big( V(z, \theta) - \int_0^z V(z-y)dF(y) + F(z) -1 \Big)dz \Bigg) .
    \end{split}
\end{equation}

To prove the proposition, we will study equations of the general form:
\begin{equation}
\label{eq:alpha_prove_2}
    \begin{split}
        u(x) = \alpha \Bigg( g(x) + \int_0^x \Big( u(z) - \int_0^z u(z-y)dF(y) \Big)dz \Bigg) .
    \end{split}
\end{equation}

We introduce the operator $\Psi$, acting on measurable locally bounded functions $h:[0, +\infty] \mapsto \mathbbm{R}$, as:
\begin{equation}
\label{eq:IntegralOperator}
(\Psi h)(x) = \int_0^x \Big( h(z) - \int_0^z h(z-y)dF(y) \Big)dz, \quad x \geq 0 
\end{equation}
Notice that the transformation $h \mapsto \Psi h$ is linear and for every $h$, $\Psi h : [0, + \infty[ \mapsto \mathbbm{R}$ is continuous, hence measurable and locally bounded. Thus, powers of the operator $\Psi$ are defined in the usual way.
\begin{equation*}
\Psi^0h = h, \qquad \Psi^nh = \Psi (\Psi^{n-1}h), \quad  n \in \mathbb N.
\end{equation*}

Let, $\norm{h}_{[0,x]} = \sup_{z \in [0,x]} |h(z)| $. Then:
\begin{equation*}
    \begin{split}
    | (\Psi h)(x) |  &\leq \int_0^x \Big( |h(z)| +\int_0^z |h(z-y)|dF(y) \Big)dz \leq 2x \norm{h}_{[0,x]}.
    \end{split} 
\end{equation*}
If the inequality 
\begin{equation}
\label{Eq bound Psi n}
    \norm{(\Psi^n h)}_{[0,x]} \leq \frac{2^n x^n}{n!} \norm{h}_{[0,x]}
\end{equation}
holds, for some $n \in \mathbbm{N}$, then
\begin{equation*}
    \begin{split}
    | (\Psi^{n+1} h)(x) |  &\leq \int_0^x \Big( | (\Psi^{n}h)(z)| +\int_0^z |(\Psi^{n}h)(z-y)|dF(y) \Big)dz \\
     &\leq \int_0^x 2 \frac{2^n z^n}{n!} \norm{h}_{[0,x]} dz = \frac{2^{n+1} x^{n+1}}{(n+1)!} \norm{h}_{[0,x]}.
    \end{split} 
\end{equation*}
Thus, by induction, \eqref{Eq bound Psi n} holds for every $n \in \mathbb N$.
%
%
Therefore, for every $x \in [0, \infty[$, fixed, there is some $n \in \mathbb N$ such that $\Psi^n$ is a contraction in the space of measureable and bounded functions $h:[0,x] \mapsto \mathbbm{R}$. It follows from the contraction principle that equation \eqref{eq:omega_prove_1} has one unique solution. Further, $\lim_{n \to \infty} (\alpha^n \Psi^n) h = 0$, uniformly in $[0,x]$ for any given $h$ and any fixed $x \in [0, +\infty[$. 

Let $u_{\alpha, g}$ be the solution of equation \eqref{eq:alpha_prove_2}  for given $g$ and $\alpha$. Then,
\begin{equation*}
    \begin{split}
        u_{\alpha,g} &= \alpha (g + \Psi u_{\alpha,g}) = \alpha g + \alpha \Psi(\alpha(g + \Psi u_{\alpha,g})) 
        = \alpha g + \alpha^2 \Psi g + \alpha^2 \Psi u_{\alpha,g} \\ 
        &= \alpha g + \alpha^2 \Psi g + \dots +\alpha^{n+1} \Psi^n g  +\alpha^{n+2} \Psi^{n +1}  u_{\alpha,g} .
    \end{split}
\end{equation*}
Since $\lim_{n \to \infty} \alpha^n \Psi^n u_{\alpha,g}(x) = 0$, this shows that $u_{\alpha,g}$ admits the series representation:
\begin{equation*}
    \begin{split}
        u_{\alpha,g} &= \sum_{n=0}^\infty\alpha^{n+1} \Psi^n g ,
    \end{split}
\end{equation*}
which converges uniformly with respect to $\alpha$ on compact intervals. 
Thus, we can differentiate term by term and obtain
\begin{equation*}
    \begin{split}
        \frac{d}{d\alpha} u_{\alpha,g}(x) &= \sum_{n=0}^\infty (n+1)\alpha^n(\Psi^n g)(x)\\
        &= \sum_{n = 0}^\infty \alpha^n (\Psi^n g)(x) + \sum_{n = 1}^\infty n \alpha^n (\Psi^n g)(x) \\
        &= \frac{1}{\alpha}u_{\alpha,g} + \sum_{n=1}^\infty \alpha^n (\Psi^n g)(x) + \sum_{n=2}^\infty (n-1) \alpha^n (\Psi^n g)(x)\\
        &= \frac{1}{\alpha}u_{\alpha,g} + \sum_{n=0}^\infty \alpha^{n+1} (\Psi^{n+1} g)(x) + \sum_{n=1}^\infty n \alpha^{n+1} (\Psi^{n+1} g)(x) \\
        &= \frac{1}{\alpha} u_{\alpha,g} + (\Psi u_{\alpha,g})(x) + \sum_{n=1}^\infty \alpha^{n+1} (\Psi^{n+1} g)(x) + \sum_{n=2}^\infty (n-1) \alpha^{n+1} (\Psi^{n+1} g)(x)  \\
        &= \frac{1}{\alpha} u{\alpha,g} + (\Psi u_{\alpha,g})(x) + \sum_{n=0}^\infty \alpha^{n+2} (\Psi^{n+2} g)(x) + \sum_{n=1}^\infty n \alpha^{n+2} (\Psi^{n+2} g)(x)  \\
        &= \frac{1}{\alpha} u_{\alpha,g} + (\Psi u_{\alpha,g})(x) + ( \alpha \Psi^2 u_{\alpha,g})(x) + \dots + ( \alpha^{k-1} \Psi^k u_{\alpha,g})(x) + \sum_{n = 1}^\infty n \alpha^{n+k} (\Psi^{n+k} g)(x) \\
        &= \sum_{n=0}^\infty \alpha^{n-1} (\Psi^n u_{\alpha,g})(x) = \frac{1}{\alpha^2} u_{\alpha,u_{\alpha,g}}.
    \end{split}
\end{equation*}

For any $h:[0,x] \mapsto \mathbbm{R}$ locally absolutely continuous function:
\begin{equation*}
    \begin{split}
        (\Psi h)(x) &= \int_0^x \Big( h(z) - \int_0^z h(z-y)dF(y) \Big)dz \\
        &= \int_0^x \Big(h(z) - [h(z-y)F(y)]_{y = 0}^{y = z} - \int_0^z h'(z-y)F(y)dy \Big)dz \\
        &= \int_0^x (h(z) - h(0)F(z))dz - \int_0^x \int_y^x h'(z-y)F(y)dzdy \\
        &= \int_0^x (h(z) - h(0)F(z))dz - \int_0^x \big(h(x-y) - h(0) \big)F(y)dy \\
        &= \int_0^x h(z)dz - \int_0^x h(x-y)F(y)dy = \int_0^x h(z) (1+F(x-z))dz .\\
    \end{split}
\end{equation*}
Thus, $h>0$ implies $(\Psi h)>0$, which implies $(\Psi^n h)>0$, $\forall n \in \mathbbm{N}$, and therefore $u_{\alpha,h}>0$ for any $\alpha >0$. This argument shows that $\frac{d}{d\alpha} V =\frac{1}{\alpha^2} u_{\alpha,V} >0$ as $V >0$. Therefore $V$ is strictly increasing with $\alpha$.
\end{proof}

According to Proposition \ref{prop:alpha_proof}, in order to find $\theta$ minimizing the probability of ruin, it is sufficient to find $\theta$ minimizing $\frac{\mathbb E^\theta[N_1]}{c(\theta)}$. For example, using the logit demand model \eqref{eq:logit}, the optimal loading is found with direct differentiation of $\alpha$ and is given by:
\begin{equation}
\label{eq:alpha_direct}
    \begin{split}
\theta_{ruin} = \frac{1}{\beta_1} \Big( \ln \big( \frac{\lambda \E[Y]}{r \beta_1}  \big) - \beta_0 \Big).
    \end{split}
\end{equation}
However, the loading that maximizes the expected profit is:
\begin{equation*}
    \begin{split}
    \theta_{profit} = \argmax_\theta 
        \E^\theta[X_1 \given[] X_0 = x] = \argmax_\theta \{ \theta\E^\theta[N_1]\E[Y] -r\},
    \end{split}
\end{equation*} 
which is, in the case of logit demand \eqref{eq:logit}, the unique solution of:
\begin{equation}
\label{eq:profit_equation}
    \begin{split}
    1 + e^{\beta_0 + \beta_1 \theta} - \beta_1 \theta e^{\beta_0 + \beta_1 \theta} = 0.
    \end{split}
\end{equation}
Thus, in general, $\theta_{ruin}$ does not coincide with $\theta_{profit}$.

\section{The Multiple Risk Case }
\label{section:companyRisk}

In this section, we explore how dependencies between risks available in an insurance market translate into risk exposure for a company through its market shares on the different risks. It turns out that this mechanism is non trivial when the risks are dependent. For the sake of simplicity, we assume that the company offers insurance for two risks in a market constituted by identical individuals, all of them exposed to both risks. Using the notation in equations \eqref{eq:2DimClaimProcess} and \eqref{eq:bracketSum} to denote the market claim process, $S_t = (S_t^{(1)}, S_t^{(2)})$ is the vector of the total (accumulated) amount of claims of each risk that occurred in the market, up to time $t$. The marginal distributions of $S^{(1)}$ and $S^{(2)}$ are characterized by claim intensities $\lambda_1$ and $\lambda_2$ and the severity distributions $F_{Y^{(1)}}$, $F_{Y^{(2)}}$ and their dependency structure is characterized by a parameter $\lambda^\parallel \in [0, \min(\lambda_1, \lambda_2)]$ and a joint distribution $F_{(Y^{1\parallel},Y^{2\parallel})}$, as explained in Section  $\ref{section:min_ruin}$.

\subsection{Risk Exposure as a Function of Market Shares}
\label{section:sub4.1}

To extend the demand model outlined in Section \ref{section:demand} to a market with multiple risks where the acquisition of insurance for different risks may not be independent, we propose the following interpretation for the function $p$.
\smallskip

Let $(\theta_1, \theta_2)$ be the loadings charged by the company for each risk. We assume that every  individual in the market (a potential client) is provided with a vector of bid prices ($b_1$, $b_2$). The client acquires the insurance for risk $i$ if $b_i \geq \theta_i$ (for convenience, we consider prices net of the pure premium). The distribution of the price vectors in the market is modelled by a random vector $B = (B_1, B_2)$. Thus, $p_i(\theta) = p_i(\theta_i) = \Prob \big( B_i \geq \theta_i \big)$ is the company's market share for the insurance of risk $i$ at equilibrium, given the loadings $\theta = ( \theta_1, \theta_2)$. Let $p^{(1,0)}$ be the proportion of individuals in the market holding a policy for risk 1 and no policy for risk 2. Similarly, $p^{(0,1)}(\theta)$ and $p^{(1,1)}(\theta)$ denote the proportion of individuals holding a policy only for risk 2 and for both risks, respectively. If the acquisition of polices for different risks is independent, then:
\begin{equation}
\label{eq:AcquisitionProbabilities}
    p^{(1,1)}(\theta) = p_1(\theta_1)p_2(\theta_2), \quad p^{(1,0)}(\theta) = p_1(\theta_1)(1-p_2(\theta_2)), \quad p^{(0,1)}(\theta) = p_2(\theta_2)(1-p_1(\theta_1)).
\end{equation}

Dependency between the acquisition of different risks can be introduced by considering dependent bid prices $B = (B_1, B_2)$. In particular, if the joint distribution of $B$ is characterized by an ordinary copula $C: [0,1]^2 \mapsto [0,1]$, then, according to Sklar's theorem $F_B(\theta_1, \theta_2) = C(F_{B_1}(\theta_1), F_{B_2}(\theta_2))$ \cite{nelsen2007introductionCopulas}. This gives:
\begin{equation*}
\label{eq:EQ8}
    \begin{split}
        p^{(1,0)} &= F_{B_2}(\theta_2^{-}) - C(F_{B_1}(\theta_1^{-}), F_{B_2}(\theta_2^{-})), \\
        p^{(0,1)} &= F_{B_1}(\theta_1^{-}) - C(F_{B_1}(\theta_1^{-}), F_{B_2}(\theta_2^{-})), \\
        p^{(1,1)} &=  1 -F_{B_1}(\theta_1^{-})- F_{B_2}(\theta_2^{-}) + C(F_{B_1}(\theta_1^{-}), F_{B_2}(\theta_2^{-})). \\
    \end{split}
\end{equation*}

Under this model, the company's surplus process is:
\begin{equation}
\label{eq:agg_company_claims}
    \Tilde{X}_t = u^{(1)} + u^{(2)} + \big( c^{(1)}(\theta_1) + c^{(2)}(\theta_2) \big)t - \sum_{i = 0}^{\Tilde{N}_t^{1\perp}} \Tilde{Y}_i^{1\perp} - \sum_{i = 0}^{\Tilde{N}_t^{2\perp}} \Tilde{Y}_i^{2\perp} - \sum_{i = 0}^{\Tilde{N}_t^{\parallel}} \Big(Y_i^{1\parallel} + Y_i^{2\parallel}  \Big),
\end{equation}
where $\Tilde{N}_t^{1\perp}$, $\Tilde{N}_t^{2\perp}$ , and $\Tilde{N}_t^{\parallel}$ count the number of claims received by the company concerning only risk 1, only risk 2, and both risks, respectively. Their intensities are, respectively,
\begin{equation*}
    \begin{split}
        \Tilde{\lambda}_1^\perp &= p^{(1,0)}(\theta) \big( \lambda_1^\perp + \lambda^\parallel \big) + p^{(1,1)}(\theta) \lambda_1^\perp = p_1(\theta_1) \lambda_1^\perp + p^{(1,0)}(\theta) \lambda^\parallel, \\
        \Tilde{\lambda}_2^\perp &= p_2(\theta_2) \lambda_2^\perp + p^{(0,1)}(\theta) \lambda^\parallel, \\
        \Tilde{\lambda}^\parallel &= p^{(1,1)}(\theta) \lambda^\parallel.
    \end{split}
\end{equation*}

The distribution of the single risk claim amounts $\Tilde{Y}^{1\perp}$ (resp., $\Tilde{Y}^{2\perp}$) is a mixture of the distributions $Y^{1\perp}$ and $Y^{1\parallel}$ (resp., $Y^{2\perp}$ and $Y^{2\parallel}$):
\begin{equation*}
    \begin{split}
     &F_{\Tilde{Y}^{1\perp}} = \frac{p_1\lambda_1^\perp }{p_1 \lambda_1^\perp + p^{(1,0)} \lambda^\parallel } F_{Y^{1\perp}} + \frac{p^{(1,0)} \lambda^\parallel }{p_1 \lambda_1^\perp + p^{(1,0)} \lambda^\parallel } F_{Y^{1\parallel}} \\
     &F_{\Tilde{Y}^{2\perp}} = \frac{p_2 \lambda_2^\perp }{p_2 \lambda_2^\perp + p^{(0,1)} \lambda^\parallel } F_{Y^{2\perp}} + \frac{p^{(0,1)} \lambda^\parallel }{p_2 \lambda_2^\perp + p^{(0,1)} \lambda^\parallel } F_{Y^{2\parallel}}
    \end{split}
\end{equation*}
This is because some customers insure risk 1, but not risk 2 and vice-versa. Therefore, the aggregate process for the insurer is
\begin{equation}
\label{eq:indp_company_claim}
    \Tilde{X}_t = u^{(1)} + u^{(2)} + \big( c^{(1)}(\theta_1) + c^{(2)}(\theta_2) \big)t - \sum_{i = 0}^{\Tilde{N}_t} \Tilde{Y}_i,
\end{equation}
where $\Tilde{N}_t$ is a Poisson process with intensity
\begin{equation}
\label{eq:lambda_tilde}
    \Tilde{\lambda} = p_1 \lambda_1^\perp + p_2 \lambda_2^\perp + \big( p^{(1,0)} + p^{(0,1)}+ p^{(1,1)} \big)\lambda^\parallel = p_1 \lambda_1 + p_2 \lambda_2 - p^{(1,1)}\lambda^\parallel ,
\end{equation}
and $\Tilde{Y}_i$, $i \in \mathbb N$ are i.i.d random variables with distribution
\begin{equation}
\label{eq:DepDist}
\begin{split}
    F_{\Tilde{Y}} &= \frac{p_1 \lambda_1^{\perp} }{\Tilde{\lambda}} F_{Y^{1 \perp}} + \frac{p_2 \lambda_2^{\perp} }{\Tilde{\lambda}} F_{Y^{2 \perp}} 
    + \frac{p^{(1,0)} \lambda^{\parallel} }{\Tilde{\lambda}} F_{Y^{1 \parallel}}
    + \frac{p^{(0,1)} \lambda^{\parallel} }{\Tilde{\lambda}} F_{Y^{2 \parallel}}
    + \frac{p^{(1,1)} \lambda^{\parallel} }{\Tilde{\lambda}} F_{Y^{1 \parallel}+ Y^{2 \parallel}} \\
    &=\frac{1}{p_1 \lambda_1 + p_2 \lambda_2 - p^{(1,1)} \lambda^\parallel} \bigg( p_1 \lambda_1 F_{Y^{1}} + p_2 \lambda_2 F_{Y^{2}} + p^{(1,1)}\lambda^\parallel \Big( F_{Y^{1\parallel} + Y^{2\parallel}} - F_{Y^{1\parallel}} - F_{Y^{2\parallel}} \Big) \bigg). 
    \end{split}
\end{equation}

Thus, if the risks in the market are independent (i.e. if $\lambda^\parallel = 0$), then the risk in the company's portfolio is just a sum of the risks $S^{(1)}$ and $S^{(2)}$, weighted by the respective market shares, $p_1$ and $p_2$, irrespective of any dependency between sales of policies for different risks. However, if the risks in the market are dependent ( $\lambda^\parallel \neq 0$), then the company's risk is not, in general, a weighted sum of $S^{(1)}$ and $S^{(2)}$. Further, this effect persists even in the case where sales of different policies are independent (i.e., $p^{(1,1,)} = p_1 p_2$).  On the other hand, equalities \eqref{eq:lambda_tilde} and \eqref{eq:DepDist} show that in the (unlikely) situation where clients always buy insurance for only one risk, the risk exposure of the insurer is accurately computed using only the marginal distributions of each risk (i.e. assuming that the risks are independent). This is due to the static nature of our model. For example, it does not take into account the possibility of external factors changing the frequency of claim events in both risks simultaneously.

\subsection{The Impact of Dependencies on Ruin Probability}

From the discussion above and Proposition \ref{prop:alpha_proof}, it follows that the ruin probability of a company with market shares ($p_1$, $p_2$, $p^{(1,1)}$) solves the equation 
\begin{align}
\label{eq:EQ3}
\frac{dV(x)}{dx} = &
\frac{\Tilde{\lambda}}{c^{(1)} + c^{(2)}} \Big( V(x) - \int_0^x V(x-y)dF_{\Tilde{Y}}(y) + F_{\Tilde{Y}}(x) -1 \Big),
\\ \label{eq:EQ4}
    V(0^{+}) =&
    \frac{\Tilde{\lambda}}{c^{(1)} + c^{(2)}} \E[\Tilde{Y}] ,
\end{align}
with $\Tilde{\lambda}$ and $F_{\Tilde{Y}}$ given by equations \eqref{eq:lambda_tilde} and \eqref{eq:DepDist}.
\smallskip

Since estimating the dependency structure may pose substantial difficulties, we may wish to have an a-priori bound for the error introduced by neglecting dependencies, that is, by substituting the probability $V_{ind}(x)$ for $V(x)$, where $V_{ind}(x)$ solves the equation.
\begin{equation}
    \label{eq:EQ5}
    \frac{dV(x)}{dx} = \frac{\hat{\lambda}}{c^{(1)} + c^{(2)}} \Big( V(x) - \int_0^x V(x-y)dF_{\hat{Y}}(y) + F_{\hat{Y}}(x) -1 \Big) ,
\end{equation}
where $\hat{\lambda} = \lambda_1p_1 + \lambda_2p_2$ and $F_{\hat{Y}}(x) = \frac{\lambda_1 p_1 F_{Y^{(1)}} + \lambda_2 p_2 F_{Y^{(2)}}}{\hat{\lambda}}$. Notice that $\hat{\lambda} \E[\hat{Y}] = \Tilde{\lambda} \E[\Tilde{\lambda}]$ and therefore the boundary condition for \eqref{eq:EQ5} is again \eqref{eq:EQ4}.
\smallskip

The discussion in Subsection \ref{section:sub4.1} shows that the difference $V(x) - V_{ind}(x)$ is expected to be small when $p^{(1,1)}$ is small compared to $p_1 + p_2$. The following proposition gives a precise meaning for this statement.

\begin{proposition}
\label{prop:P1}
With the notation above:
\begin{equation*}
    |V(x) - V_{ind}(x)| \leq p^{(1,1)} \lambda^\parallel \frac{e^{\frac{2\Tilde{\lambda}x}{c^{(1)} + c^{(2)}}} - 1}{\Tilde{\lambda}}
\end{equation*}
for every amount of initial reserve $x \geq 0$.
\end{proposition}

\begin{proof}
From equalities \eqref{eq:EQ3}, \eqref{eq:EQ4} and \eqref{eq:EQ5}, straightforward computations yield:
\begin{equation}
\label{eq:EQ6}
    \begin{split}
        V(x) -V_{ind}(x) &= \frac{p^{(1,1)} \lambda^\parallel}{c^{(1)} + c^{(1)}} \Bigg(  \int_0^x V_{ind}(z) - \int_0^z V_{ind}(z-y)dF_{Y^{1\parallel} + Y^{2\parallel}}(y) + F_{Y^{1\parallel} + Y^{2\parallel}}(z) -1 dz -  \\
        & \quad \int_0^x V_{ind}(z) - \int_0^z V_{ind}(z-y)dF_{Y^{1\parallel}}(y) + F_{Y^{1\parallel} }(z) -1 dz - \\
        & \quad \int_0^x V_{ind}(z) - \int_0^z V_{ind}(z-y)dF_{Y^{2\parallel}}(y) + F_{Y^{2\parallel}}(z) -1 dz -\Bigg)  \\
        & \quad + \frac{\Tilde{\lambda}}{c^{(1)} + c^{(2)}} \int_0^x (V - V_{ind})(z) - \int_0^z (V - V_{ind})(z-y)dF_{\Tilde{Y}}(y)dz
    \end{split}
\end{equation}

It can be checked that for every distribution function $G:[0, +\infty[ \mapsto [0,1]$,
\begin{equation*}
    -x \leq \int_0^x V_{ind}(z) -\int_0^z V_{ind}(z-y)dG(y) + G(z) -1dz \leq 0
\end{equation*}

Therefore, \eqref{eq:EQ6} implies:
\begin{equation*}
    \max_{y \in [0,x]} |V(x) - V_{ind}(x)| \leq \frac{ p^{(1,1)} \lambda^\parallel}{c^{(1)} + c^{(2)}}2x + \frac{\Tilde{\lambda}}{c^{(1)} + c^{(2)}} \int_0^x 2  \max_{y \in [0,z]} |V(x) - V_{ind}(y)|dz
\end{equation*}
Thus, the result follows by Grönwall's inequality \cite{dragomir2003some}.
\end{proof}

\subsection{The Impact of Dependencies on Small Companies}

Now, we proceed with the argument above to explore how dependencies affect companies of different size. We measure the size of the company by it's expected total value of claims, $\Tilde{\lambda} \E[\Tilde{Y}]$ and, to make comparisons meaningful, we consider that the total revenue is proportional to the company's size, i.e.
\begin{equation*}
    c^{(1)} + c^{(2)} = (1 + \theta) \Tilde{\lambda} \E [\Tilde{Y}], \quad \textrm{with } \theta>0 \textrm{ constant.}
\end{equation*}
Similarly, we consider the initial reserve to be proportional to size, i.e.:
\begin{equation*}
    x = x_o  \Tilde{\lambda} \E [\Tilde{Y}], \quad \textrm{with } x_0>0 \textrm{ constant.}
    \end{equation*}
Notice that, due to equations \eqref{eq:EQ3}, \eqref{eq:EQ4} and \eqref{eq:EQ5}, the effect of dependencies must be bounded in the sense that
\begin{equation*}
    |V(x_0\Tilde{\lambda}\E[\Tilde{Y}]) - V_{ind}(x_0\Tilde{\lambda}\E[\Tilde{Y}])| \leq K_1 x_0\Tilde{\lambda}\E[\Tilde{Y}] \leq K_2(p_1 + p_2) ,
\end{equation*}
for some constants $K_1, K_2 < +\infty$. However, we can use Proposition \ref{prop:P1} to obtain a better estimate:
\begin{equation}
    \label{eq:EQ7}
    |V(x_0\Tilde{\lambda}\E[\Tilde{Y}]) - V_{ind}(x_0\Tilde{\lambda}\E[\Tilde{Y}])| \leq p^{(1,1)} \lambda^\parallel \frac{e^{\frac{2\Tilde{\lambda}\frac{x_0}{1 + \theta}}{c^{(1)} + c^{(2)}}} - 1}{\Tilde{\lambda}}.
\end{equation}

Notice that the right-hand side of \eqref{eq:EQ7} has the same asymptotic behaviour as 
\begin{equation*}
    p^{(1,1)} \lambda^\parallel \frac{x_0}{1 + \theta}, \quad \textrm{when } p_1 + p_2 \to 0.
\end{equation*}
Further, if the sales of policies for different risks to the same individual are independent, then $p^{(1,1)} = p_1 p_2$ goes to zero faster than $\Tilde{\lambda} \E[\Tilde{Y}] = p_1 \lambda_1 \E[Y^{(1)}] + p_2 \lambda_2 \E[Y^{(2)}]$, when $p_1 + p_2 \to 0$. Thus, a small company selling policies for different risks independently is relatively immune to the effects of dependencies between the risks, contrary to a large company (it is obvious that a monopolistic company is fully exposed to the dependencies between risks). 
This immunity to risk's dependencies may persist even when sales of policies for different risks are not independent, provided the dependency in sales is  sufficiently mild. For example, $\lim_{p_1 + p_2 \to 0} \frac{p^{(1,1)}}{p_1 + p_2} = 0$ if the dependency between sales is modelled by a Clayton or a Frank copula in \eqref{eq:EQ8}. However, small companies are not specially protected from risk dependencies if the dependency between sales is modelled by a Pareto or a Gumbel copula.

\subsection{Optimal Loadings and Market Shares}

Since the right-hand sides of equalities \eqref{eq:EQ3} and \eqref{eq:EQ4} depend on the loadings through both $\frac{\E[\Tilde{N}_1}{c^{(1)} + c^{(2)}}$ and $F_{\Tilde{Y}}$, Proposition \ref{prop:alpha_proof} can not be generalized to models with multiple risks. However, it is possible to provide optimality conditions for the loadings $\theta = (\theta_1, \theta_2)$ minimizing the ruin probability.
\smallskip

To do this, we extend the notation introduced in the proof of Proposition \ref{prop:alpha_proof}. For any distribution function $G:[0, +\infty[ \mapsto [0,1]$, we consider the compounding operator of type \eqref{eq:IntegralOperator}
\begin{equation*}
    (\Psi_G h)(x) = \int_0^x \Big(  h(z) - \int_0^z h(z-y) d\theta(y)\Big)dz, \quad x \geq 0.
\end{equation*}
Thus, the 2-risk version of equation \eqref{eq:omega_prove_1} can be written as 
\begin{equation}
\label{eq:EQ10}
    V_\theta(x) = \frac{\Tilde{\lambda}_\theta}{c(\theta)} \bigg( \int_x^\infty 1- F_\theta(z)dz + \Big(\Psi_{F_\theta} V_\theta \Big)(x)  \bigg) ,
\end{equation}
where
\begin{equation*}
    F_\theta = \frac{\lambda_1^\perp}{\Tilde{\lambda}_\theta} p_1(\theta) F_{Y^{1\perp}} + \frac{\lambda_2^\perp}{\Tilde{\lambda}_\theta} p_1(\theta) F_{Y^{2\perp}} + \frac{\lambda^\parallel}{\Tilde{\lambda}_\theta} p^{(1,0)}(\theta) F_{Y^{1\parallel}} + \frac{\lambda^\parallel}{\Tilde{\lambda}_\theta} p^{(0,1)}(\theta) F_{Y^{2\parallel}} + \frac{\lambda^\parallel}{\Tilde{\lambda}_\theta} p^{(1,1)}(\theta) F_{Y^{1\parallel} + Y^{2\parallel}} .
\end{equation*}
Since $\frac{\lambda_1^\perp}{\Tilde{\lambda}_\theta} p_1(\theta) + \frac{\lambda_2^\perp}{\Tilde{\lambda}_\theta}p_2(\theta) + \frac{\lambda^\parallel}{\Tilde{\lambda}_\theta} p^{(1,0)}(\theta) + \frac{\lambda^\parallel}{\Tilde{\lambda}_\theta} p^{(0,1)}(\theta) + \frac{\lambda^\parallel}{\Tilde{\lambda}_\theta} p^{(1,1)}(\theta) = 1$, \eqref{eq:EQ10} becomes
\begin{equation*}
    \begin{split}
        V_{\theta}(x) &= \frac{\lambda_1^\perp p_1(\theta)}{c(\theta)} \bigg( \int_x^\infty 1- F_{Y^{1\perp}}(z)dz + (\Psi_{F_{Y^{1\perp}}}V_\theta)(x) \bigg) \\
        & \quad + \frac{\lambda_2^\perp p_2(\theta)}{c(\theta)} \bigg( \int_x^\infty 1- F_{Y^{2\perp}}(z)dz + (\Psi_{F_{Y^{2\perp}}}V_\theta)(x) \bigg) \\
        & \quad + \frac{\lambda^\parallel p^{(1,0)}(\theta)}{c(\theta)} \bigg( \int_x^\infty 1- F_{Y^{1\parallel}}(z)dz + (\Psi_{F_{Y^{1\parallel}}}V_\theta)(x) \bigg) \\
        & \quad + \frac{\lambda^\parallel p^{(0,1)}(\theta)}{c(\theta)} \bigg( \int_x^\infty 1- F_{Y^{2\parallel}}(z)dz + (\Psi_{F_{Y^{2\parallel}}}V_\theta)(x) \bigg) \\
        & \quad + \frac{\lambda^\parallel p^{(1,1)}(\theta)}{c(\theta)} \bigg( \int_x^\infty 1- F_{Y^{1\parallel} +Y^{2\parallel}}(z)dz + (\Psi_{F_{Y^{1\parallel} + Y^{2\parallel}}}V_\theta)(x) \bigg) .
    \end{split}
\end{equation*}

We write this in abbreviated form:
\begin{equation*}
    V_{\theta}(x) = <\alpha(\theta), \Gamma(x)> + (<\alpha(\theta), \Psi>V_\theta)(x) ,
\end{equation*}
where $\alpha(\theta)$ is the vector
\begin{equation*}
    \alpha(\theta) = \frac{1}{c(\theta)} \Big(\lambda_1^\perp p_1(\theta), \lambda_2^\perp p_2(\theta), \lambda^\parallel p^{(1,0)}, \lambda^\parallel p^{(0,1)}, \lambda^\parallel p^{(1,1)} \Big),
\end{equation*}
$\Gamma(x)$ is the vector function
\begin{equation*}
    \Gamma(x) = \Big(\int_x^\infty 1- F_{Y^{1\perp}}(z)dz, \int_x^\infty 1- F_{Y^{2\perp}}(z)dz, \int_x^\infty 1- F_{Y^{1\parallel}}(z)dz, \int_x^\infty 1- F_{Y^{2\parallel}}(z)dz, \int_x^\infty 1- F_{Y^{1\parallel} + Y^{2\parallel}}(z)dz \Big),
\end{equation*}
$\Psi$ is the vector of operators
\begin{equation*}
    \Psi = \Big(\Psi_{F_{Y^{1\perp}}}, \Psi_{F_{Y^{2\perp}}}, \Psi_{F_{Y^{1\parallel}}}, \Psi_{F_{Y^{2\parallel}}}, \Psi_{F_{Y^{1\parallel}} + Y^{2\parallel}} \Big) ,
\end{equation*}
and $<\cdot, \cdot>$ is the usual inner product in $\mathbbm{R}^5$. 
\smallskip

Using the argument in the proof of Proposition \ref{prop:alpha_proof}, we see  that $V_\theta$ admits the series representation
\begin{equation*}
    V_\theta(x) = \sum_{n=0}^\infty \Big(<\alpha(\theta), \Psi>^n <\alpha(\theta), \Gamma> \Big)(x) .
\end{equation*}
Similarly, any vector $\gamma \in \mathbbm{R}^5$ and any bounded measurable function $g:[0, +\infty[ \mapsto \mathbbm{R}$ define one unique function
\begin{equation*}
    u_{\gamma, g} (x) = \sum_{n = 0}^\infty \Big(<\gamma, \Psi>^n g\Big)(x) .
\end{equation*}
This function is analytic with respect to $\gamma$, with partial derivatives
\begin{equation*}
    \frac{\partial{u_{\gamma, g}}}{\partial{\gamma_i}} = \sum_{n = 0}^\infty <\gamma, \Psi>^n\big( \Psi_i u_{\gamma,g}\big) = u_{\gamma, \Psi_i u_{\gamma, g}}, \quad i = 1, \dots, 5 .
\end{equation*}

Taking into account the chain rule for derivatives, this proves the following proposition.

\begin{proposition}
\label{prop:P2}
If $\theta \mapsto \alpha(\theta)$ is differentiable, then $\theta \mapsto V_\theta(x)$ is differentiable for every $x \geq 0$ and

$$\frac{\partial{}}{\partial{\theta_i}} V_\theta(x) = \sum_{j = 1}^5 u_{\alpha(\theta),(\Gamma_j + \Psi_j V_\theta)} \frac{\partial{\alpha_j(\theta)}}{\partial{\theta_i}}, \quad i = 1,2.$$
\end{proposition}

By Proposition \ref{prop:P2}, the optimal loadings satisfy the equation 
\begin{equation}
    \label{eq:EQ11}
     \sum_{j = 1}^5 u_{\alpha(\theta),(\Gamma_j + \Psi_j V_\theta)} \frac{\partial{\alpha_j(\theta)}}{\partial{\theta_i}} = 0, \quad i = 1,2
\end{equation}

Contrary to the single-risk case, the odds of finding explicit solutions for this equation seem very low, even in simple cases. 
However, \eqref{eq:EQ11} can be numerically solved by Newton's algorithm, the second-order partial derivatives being
\begin{equation*}
    \frac{\partial^2}{\partial{\theta_i}\partial{\theta_j}} V_\theta(x) = \sum_{k = 1}^5 u_{\alpha(\theta),(\Gamma_k + \Psi_k V_\theta)} \frac{\partial^2{\alpha_k(\theta)}}{\partial{\theta_i}\partial{\theta_j}} + \sum_{k = 1}^5\sum_{l = 1}^5 u_{\alpha(\theta),\Psi_k} u_{\alpha(\theta),(\Gamma_l + \Psi_l V_\theta)} \frac{\partial{\alpha_k(\theta)}}{\partial{\theta_i}}  \frac{\partial{\alpha_l(\theta)}}{\partial{\theta_j}} .
\end{equation*}

Notice that the expected profit is
\begin{equation*}
    c^{(1)}(\theta) + c^{(2)}(\theta) - \Tilde{\lambda} \E[\Tilde{Y}] = \theta_1 p_1(\theta_1)\lambda_1\E[Y^{(1)}] + \theta_2 p_2(\theta_2)\lambda_2\E[Y^{(2)}].
\end{equation*}
Thus, it depends only on the marginal distribution of the claim processes $S^{(1)}$, $S^{(2)}$, being independent of the dependency structure. It follows that the loadings minimizing the joint profit coincide with the loadings minimizing the profit on each risk, separately.  That is, a pricing strategy that completely focus on expected profit completely fails to take both dependencies between risks and dependencies between sales of policies into account.


\section{Numerical Results}
\label{seq:numerical}

Throughout this section, $Y_i^{(i)}$ are assumed to be i.i.d gamma distributed random variables with shape parameter, $a^{(i)}$, and scale parameters, $k^{(i)}$, which means that the mean is, $\E[Y^{(i)}] = a^{(i)}k^{(i)}$, for $i = 1,2$. In the following numerical analysis let $a^{(1)} = a^{(2)} = 2$, $k^{(1)} = k^{(2)} = 500$, $\lambda^{(1)} = \lambda^{(2)} = 800$, $\beta_0^{(1)} = \beta_0^{(2)} = -0.5$, $\beta_1^{(2)} = 4$ and $\beta_1^{(1)} = 4.5$. That is, the difference stems from surplus process 2 being more sensitive to the loading via the parameter $\beta_1^{(2)}$.  $r^{(i)}$ is taken to be $20\%$ of the pure premium if the exposure was $40\%$, that is $r^{(i)} = 0.4 * 0.2 k^{(i)} a^{(i)} N^{(i)}$. The operational cost is therefore $8\%$ of the expected total amount of claims in the market. The Clayton Lévy copula is considered for positive dependence and the parameter is set to $\omega = 1$. Finally, let $\theta_{ruin}^*$ and $\theta_{profit}^*$ denote the optimal loading when the ruin probability and expected profit criterion is used, respectively. The programming language R was used for every calculation. \\

\subsection{Single Surplus Process}

The surplus processes are first considered separately. The ruin probability and the expected profit is plotted as a function of $\theta$ for the two processes in Figures \ref{fig:logit_1} and \ref{fig:logit_2}. $\theta_{ruin}^*$ was found by minimizing $\alpha$.

\begin{figure}[H]
  \centering
  \includegraphics[scale = 0.4]{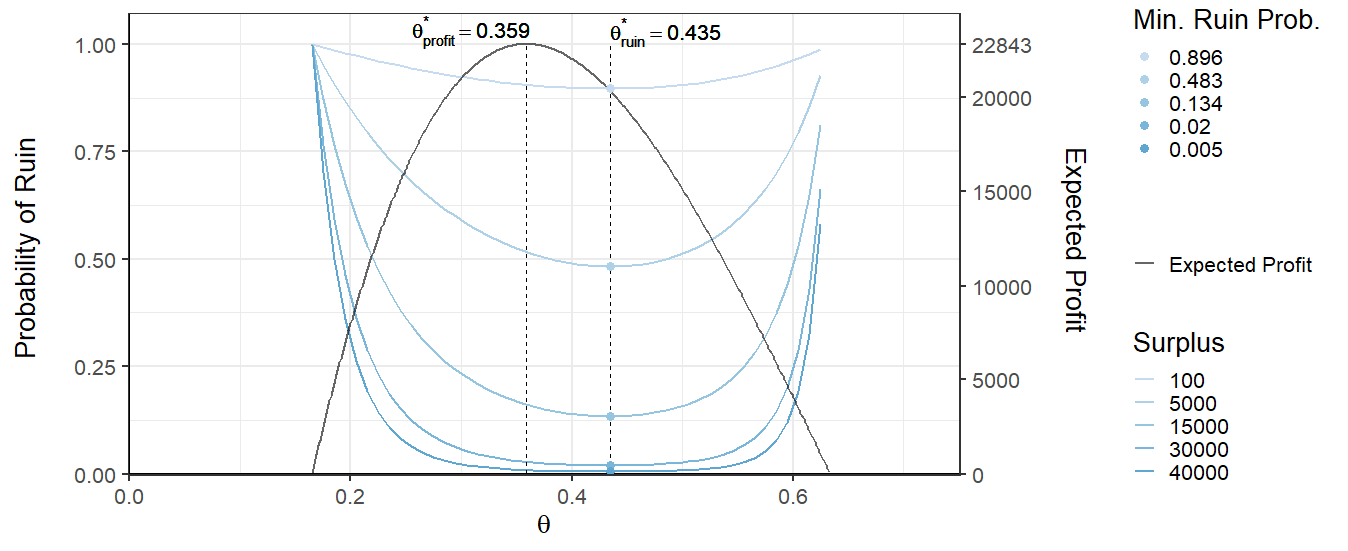}
  \caption[Optimal loading parameter of surplus process 1.]{Surplus process 1. The blue lines show the ruin probability as a function of $\theta$ for a given surplus x. The black line shows the expected profit per time unit as a function of $\theta$. The blue dots show the minimum ruin probability for each surplus. $\theta_{profit}^*$ and $\theta_{ruin}^*$ denote the optimal security loading parameter for the expected profit and for the probability of ruin, respectively. }
  \label{fig:logit_1}
\end{figure}

From Figure \ref{fig:logit_1} it can be seen that the optimal security loading parameter for the ruin probability is, $\theta_{ruin}^* = 0.435$, while the $\theta$ that maximizes the expected profit is lower, $\theta_{profit}^* = 0.359$. Moreover, in this example, the maximum expected profit is 22.843 units and is given at $\theta_{profit}^*$. The expected profit taken at the point $\theta_{ruin}^*$ is lower, close to 20.000 units. 

\begin{figure}[H]
  \centering
  \includegraphics[scale = 0.4]{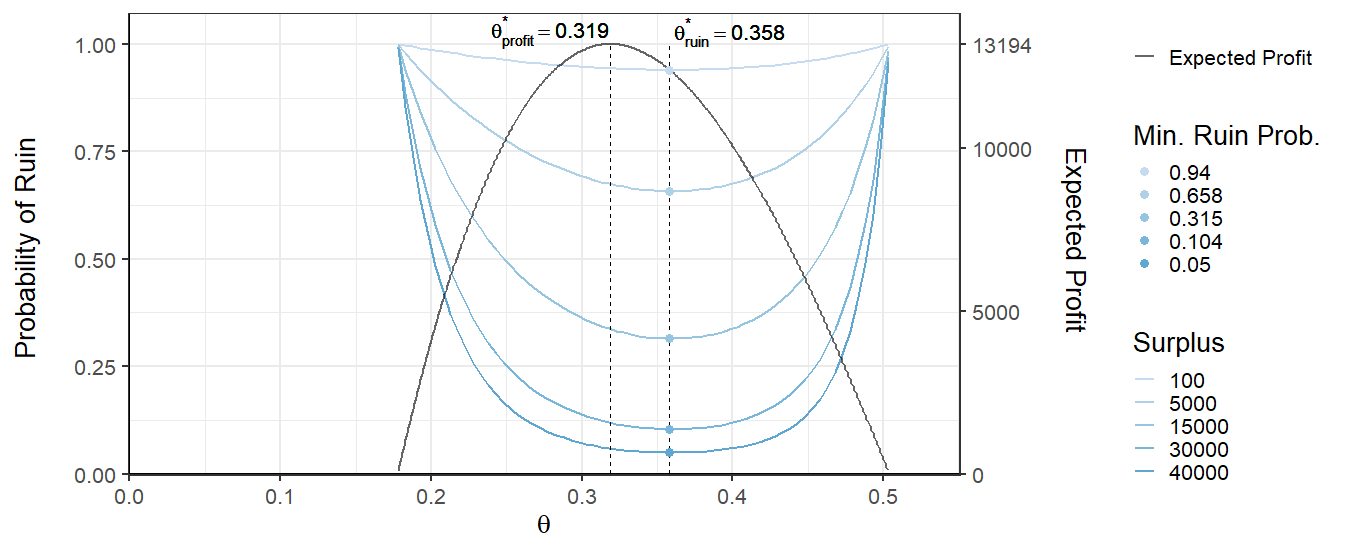}
  \caption[Optimal loading parameter of surplus process 2.]{Surplus process 2. The blue lines show the ruin probability as a function of $\theta$ for a given surplus x. The black line shows the expected profit per time unit as a function of $\theta$.  The blue dots show the minimum ruin probability for each surplus. $\theta_{profit}^*$ and $\theta_{ruin}^*$ denote the optimal security loading parameter for the expected profit and for the probability of ruin, respectively.}
  \label{fig:logit_2}
\end{figure}

From Figure \ref{fig:logit_2} it can be seen that the optimal security loading parameter for the ruin probability is $\theta_{ruin}^* = 0.358$, while the $\theta$ that maximizes the expected profit is again lower or $\theta_{profit}^*= 0.319$.  \\

Obviously, for both processes, the ruin probability decreases with increasing surplus. Moreover, it can be seen that surplus process $X_2$ has higher probability of ruin than surplus process $X_1$ for the same amount of surplus. The sensitivity of the demand curve affects the ruin probability and $\theta_{ruin}^*$ greatly. The more sensitive to the exposure the demand curve is, the closer the $\theta_{profit}^*$ and $\theta_{ruin}^*$ are. This more sensitive curve also has higher probability of ruin for a given surplus, which indicates that more competitive insurance products are riskier. These effects can be seen if the two Figures (\ref{fig:logit_1} and \ref{fig:logit_2}) are compared. Conversely, if the demand curve is not sensitive to the price, then the gap between $\theta_{profit}^*$ and $\theta_{ruin}^*$ can become quite large. Additionally, it can be seen from the curve at surplus = 100 that the ruin probability for $\theta_{profit}^*$ and $\theta_{ruin}^*$ are similar but as the surplus grows the values start to differ and once the surplus is great enough the two values $\theta_{profit}^*$ and $\theta_{ruin}^*$ result in similar ruin probabilities again. This means that if the insurance firm has high enough surplus then they can choose arbitrary $\theta$ without risking the chance of ruin. If the surplus is great enough then the value of $\theta$ does not matter as much. However, having too much reserves can be bad for insurance companies as it can be seen as a negative leverage. The bowl shape of the blue curves in the two Figures (\ref{fig:logit_1} and \ref{fig:logit_2})  is because of the interplay between the fixed cost and the demand curve. \\

$\theta_{ruin}^*$ should give the minimum ruin probability at all surplus values. This can be tested by graphing multiple ruin probability curves and compare it with the one obtained by $\theta_{ruin}^*$. Figure \ref{fig:test_many_thetas} shows that $\theta_{ruin}^*$ gives the minimum ruin probability indeed.

\begin{figure}[H]
  \centering
  \includegraphics[scale = 0.4]{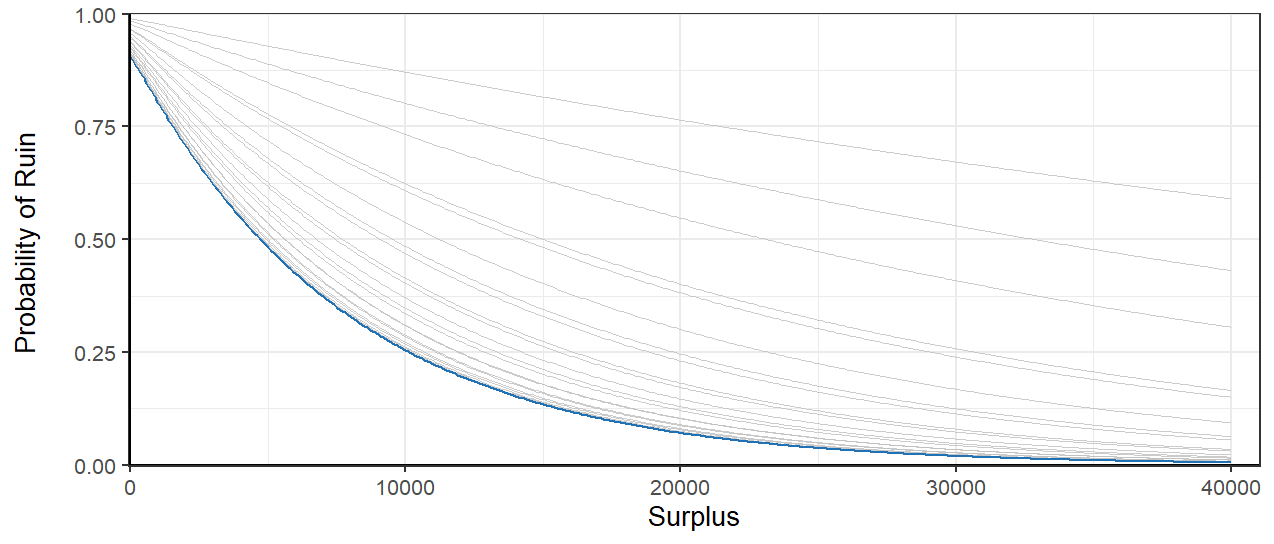}
  \caption[Optimal value function of surplus process, $X_1$.]{The figure compares the optimal value function of surplus process, $X_1$ (blue line) to other cost functions (grey lines). The blue line is achieved by setting $\theta = \theta_{ruin}^*$.  All the cost functions lie above the optimal value function, as is expected. }
  \label{fig:test_many_thetas}
\end{figure}

\subsection{Two Aggregated Surplus Processes with Common Loading}

Next, the two surplus processes, $X_1$ and $X_2$ are aggregated, both when the claims are independent and dependent. The acquisition is independent in this subsection.

\begin{figure}[H]
  \centering
  \includegraphics[scale = 0.38]{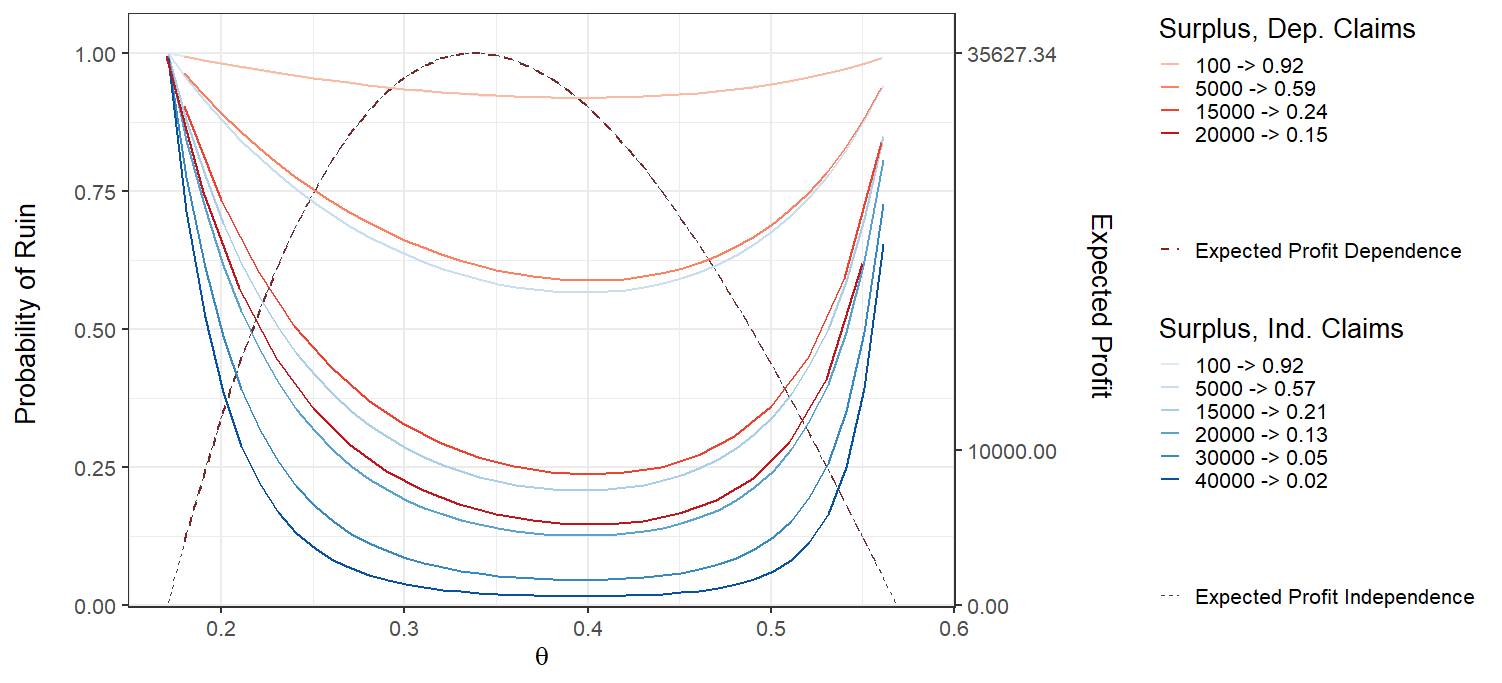}
  \caption[Surplus processes 1 and 2 aggregated. Showing the ruin probability and optimal $\theta$ both in the case of independence and dependence. ]{Ruin probability when $X_1$ and $X_2$ are aggregated as a function of the security loading parameter, $\theta$, both when they are independent and dependent via Clayton Lévy copula with $\omega = 0.5$. The blue curves show the ruin probability when the two processes are independent for different values of the surplus and the red curves show the same for the dependent case. The curves have similar shapes, but the ruin probability is higher in the case of dependence, for the same surplus. The values in the legend show the minimum ruin probability for a given surplus (surplus $\rightarrow$ probability).}
  \label{fig:logit_dep_vs_indp}

\end{figure}

Figure \ref{fig:logit_dep_vs_indp} shows the ruin probability of the aggregated surplus process as a function of the security loading parameter, $\theta$, both when they are independent and dependent via Clayton Lévy copula. The red curves represent dependence while the blue curves represent independence. \\

Firstly, it can be seen that the expected profit is the same for dependence and independence and from the figure, $\theta^*_{profit} \approx 0.34$. The reason is that the claim mean and the claim frequency is almost the same (numerically) for dependence and independence. \\

Secondly, the dependent case has a higher probability of ruin than the independent case for the same amount of surplus. However, the ruin probability is almost the same for small surplus values as can be seen from the figure. Interestingly, the optimal loading for dependence and independence seem to be the same and numerically the values are $\theta^*_{ruin,dep}  = 0.4  = \theta^*_{ruin,indp}$.  The surplus value does not change the optimal loading $\theta^*$, as expected. The reason why the ruin probability difference between the dependent and independent cases is relatively small is because of the probability $p^{(0,1)}(\theta)$. The fact that the insurance company does not always have the both claims $Y^{1\parallel}$ and $Y^{1\parallel}$ when a common jump occurs reduces the risk.
\\

Finally, the difference of the two ruin probability curves (red and blue) for a given surplus seems to be increasing with increasing surplus, meaning that the ruin probability in the independent case decreases more rapidly with increasing surplus then for the dependent case. Therefore, it is clear that the positive dependent case is riskier. \\\

Note that   $\theta_{ruin}^* \approx 0.4$, which is very close to the weighted average of the optimal loading parameter of the isolated surplus processes where the weight is the exposure ratio of each surplus process, that is

\begin{equation*}
    \theta_{weighted} = \frac{0.435 \frac{1}{1 + \exp(-0.6 + 4*0.4)} + 0.358\frac{1}{1 + \exp(-0.6 + 4.5*0.4)}}{\frac{1}{1 + \exp(-0.6 + 4*0.4)} + \frac{1}{1 + \exp(-0.6 + 4.5*0.4)}} \approx 0.4,
\end{equation*}

which strongly indicates that the optimal value, $\theta_{ruin}^*$, is simply the weighted average. \\

\subsubsection{Two Aggregated Surplus Processes with Separate Loadings}

It is more realistic to consider $\theta$ as a vector so that the loading parameter can be different for each surplus process separately, to spread the total premium over the policies in an optimal way. The two surplus processes, $X_1$ and $X_2$, are aggregated as before and the constants are the same, but let $\theta = (\theta^{(1)}, \theta^{(2)})$. \\

\begin{figure}[ht]
  \begin{subfigure}{.5\textwidth}
  \centering
  \includegraphics[width=.9\linewidth]{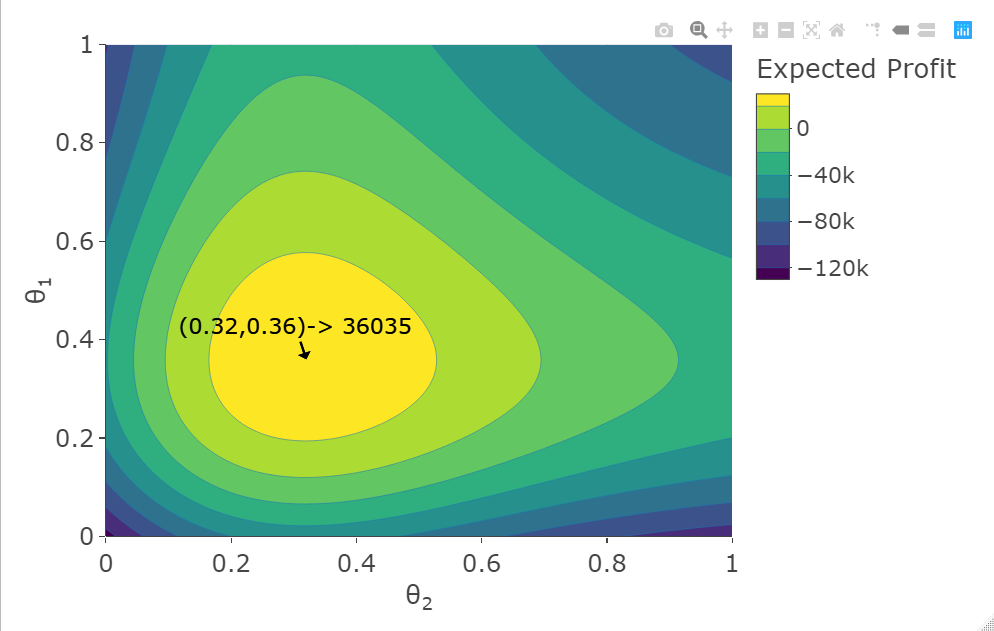}  
\end{subfigure}
\begin{subfigure}{.5\textwidth}
  \centering
  \includegraphics[width=.9\linewidth]{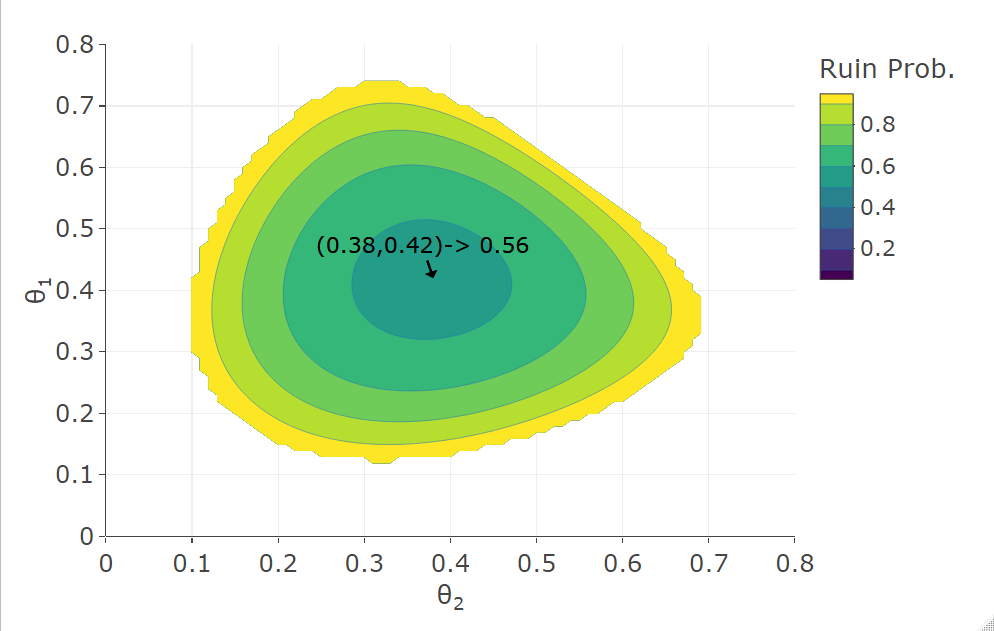}  
\end{subfigure}
  \caption[Two security loading parameters of two aggregated independent surplus processes. ]{Expected profit (left) and the ruin probability (right) when $X_1$ and $X_2$ are aggregated, as a function of the security loading parameters, $\theta^{(1)}$ and $\theta^{(2)}$. The processes are assumed to be independent and the surplus is fixed at $x = 5000$. The parenthesis in the right figure shows the optimal values of $\theta^{(1)}$ and $\theta^{(2)}$ with the ruin probability as a criterion. The arrow indicates which values $\theta^{(1)}$ and $\theta^{(2)}$ are mapped into, thus showing the minimum ruin probability. The parenthesis in the left figure shows the same for the expected profit. The shape of the contour plot is due to the fact that the $\theta$ grid considered is sparser for values that give high ruin probability.}
  \label{fig:both_two_indp_15000}
\end{figure}

Figure \ref{fig:both_two_indp_15000} shows the expected profit (left) and the ruin probability (right), when $X_1$ and $X_2$ are assumed to be independent and aggregated, as a function of the security loading parameters, $\theta^{(1)}$ and $\theta^{(2)}$. The surplus is fixed at $x = 5000$ and the optimal values are shown. It should be noted that many surplus values were tested and they all gave the same value for $\theta_{ruin}^{*(1)}$, $\theta_{ruin}^{*(2)}$, $\theta_{profit}^{*(2)}$, and $\theta_{profit}^{*(2)}$ as shown, only the ruin probability level changed. Note that the optimal loading parameters for the expected profit are the same as those for the individual surplus processes. However, the optimal loading parameters for the ruin probability change when compared to the individual one (compare it with Figures \ref{fig:logit_1} and \ref{fig:logit_2}). When compared to the optimal loading parameter for the individual surplus process, $\theta^{(1)}$ decreases from $0.435$ to $0.42$ and $\theta^{(2)}$ increases from 0.358 to 0.38. Therefore, the optimal security loading parameter decision is to decrease the loading parameter of the less sensitive surplus process while increasing the loading parameter of the more sensitive surplus process. Additionally, when compared to Figure \ref{fig:logit_dep_vs_indp}, the minimum ruin probability for one shared loading is $0.57$ while the ruin for two loadings is $0.56$, showing only a marginal difference. When the same is done for other surplus values a similar difference is found. The expected profit is marginally higher. \\

Lastly, consider the case when the surplus processes are assumed to be dependent via Lévy Clayton copula and the loadings can be different for each surplus process separately. Figure \ref{fig:both_two_dep_15000} shows the ruin probability when $X_1$ and $X_2$ are aggregated as a function of the security loading parameters, $\theta^{(1)}$ and $\theta^{(2)}$. The shape of the contour plots is due to the fact that the $\theta$ grid considered is sparser for values that give high ruin probability. The surplus is fixed at $x = 5000$. \\

\begin{figure}[ht]
  \begin{subfigure}{.5\textwidth}
  \centering
  \includegraphics[width=.96\linewidth]{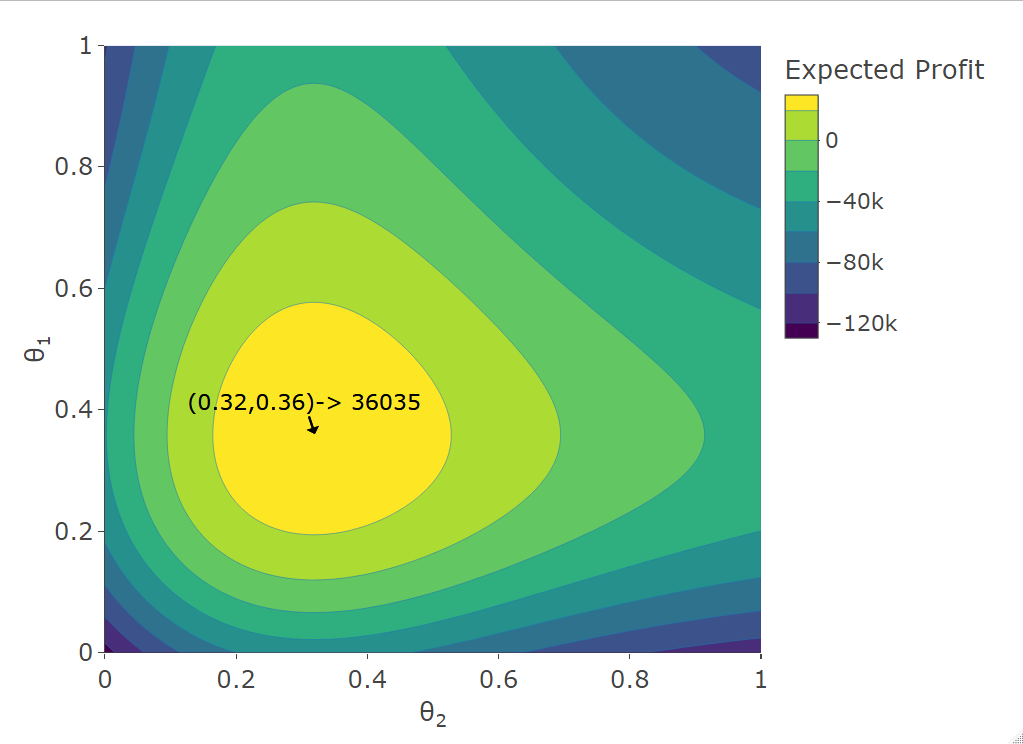}  
\end{subfigure}
\begin{subfigure}{.5\textwidth}
  \centering
  \includegraphics[width=.96\linewidth]{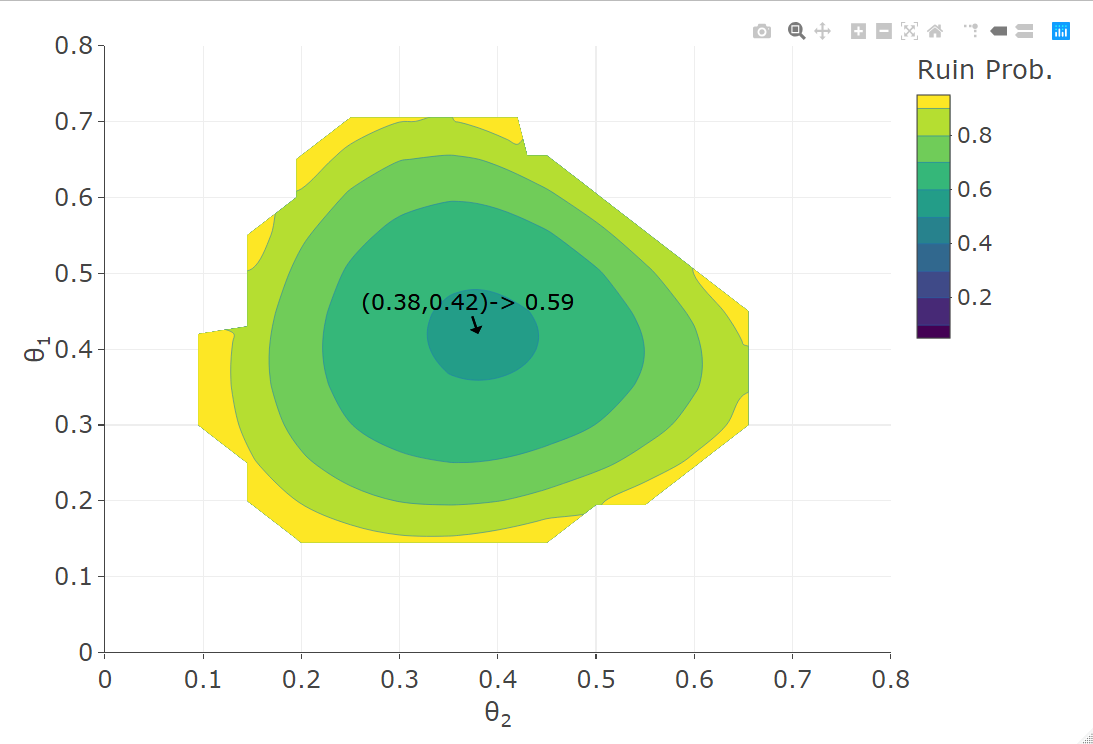}  
\end{subfigure}
  \caption[Two security loading parameters of two aggregated dependent surplus processes.]{Expected profit (left) and the ruin probability (right) when $X_1$ and $X_2$ are aggregated, as a function of the security loading parameters, $\theta^{(1)}$ and $\theta^{(2)}$. The processes are assumed to be dependent via Clayton Lévy copula and the surplus is fixed at $x = 5000$. The parenthesis in the right figure shows the optimal values of $\theta^{(1)}$ and $\theta^{(2)}$ with the ruin probability as a criterion. The arrow shows which values $\theta^{(1)}$ and $\theta^{(2)}$ are mapped into, thus showing the minimum ruin probability. The parenthesis in the left figure shows the same except for the expected profit. The shape of the contour plot is due to the fact that the $\theta$ grid considered is sparser for values that give high ruin probability.}
  \label{fig:both_two_dep_15000}
\end{figure}

It can be seen that the optimal loadings $\theta^{(1)}$ and $\theta^{(2)}$ are the same as the ones in the case of independence and the minimum ruin probability is higher (compared to the case in Figure \ref{fig:both_two_indp_15000}). Both the values and the optimal loadings of the expected profit are the same as the independent case. 
Again, the optimal security loading parameter decision is to decrease the loading parameter of the less sensitive surplus process while increasing the loading parameter of the more sensitive surplus process. The difference between the ruin probability in Figure \ref{fig:both_two_dep_15000} vs Figure \ref{fig:both_two_indp_15000} is only $0.03$ but in this case the surplus is low compared to the expected profit. If the surplus would be increased to $\approx 20.000$ the difference would become greater. The difference would then decrease again if the surplus were increased to $\approx 40.000$. \\

Additionally, when compared to Figure \ref{fig:logit_dep_vs_indp}, the minimum ruin probability for one common loading is $0.59$, which is the same as the ruin probability for separate loading selections, therefore the difference is only marginal. \\

\subsection{Dependent Claims and Dependent Acquisition}

It is time to look at the case when we have dependent claims and dependent acquisition. Note that the case when we have independent claims and dependent acquisition is the same as the total independence case. We will look both at the case when the acquisition is modelled with a Gumbel and Clayton dependency structure. To compare these two structures we use Kendell's tau. The following equations relate the copula parameters, $\omega_{clayton}$ and $\omega_{gumbel}$ to kendell's tau, $\tau$. 

\begin{equation*}
  \omega_{clayton} = \frac{2 \tau}{1-\tau}, \quad   \omega_{gumbel} = \frac{1}{1-\tau}.
\end{equation*}

We know that the expected profit is the same as before for all values of $\tau$. Therefore, we analyze the ruin probability. \\

\begin{figure}[ht]
  \begin{subfigure}{.5\textwidth}
  \centering
  \includegraphics[width=.9\linewidth]{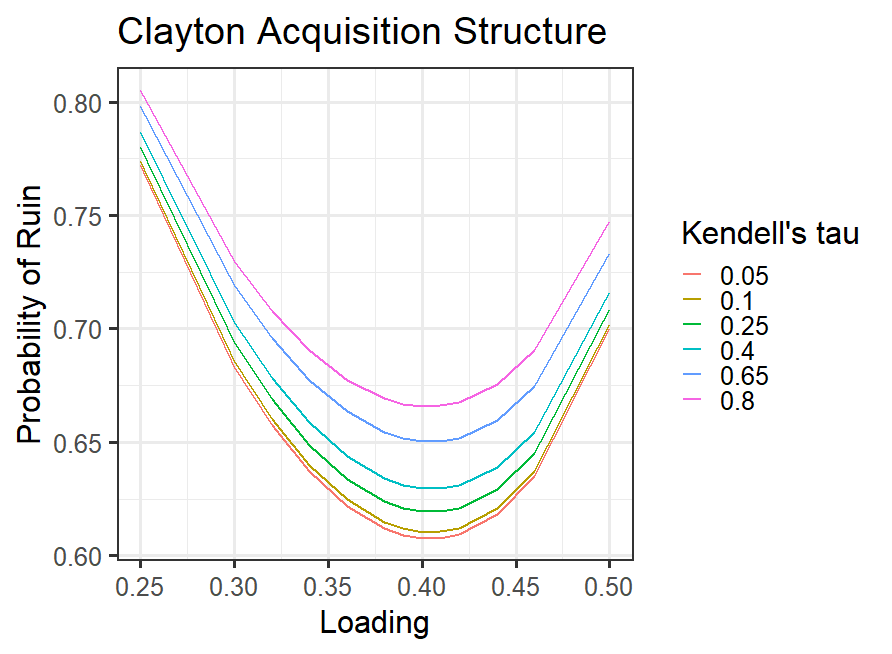}  
\end{subfigure}
\begin{subfigure}{.5\textwidth}
  \centering
  \includegraphics[width=.9\linewidth]{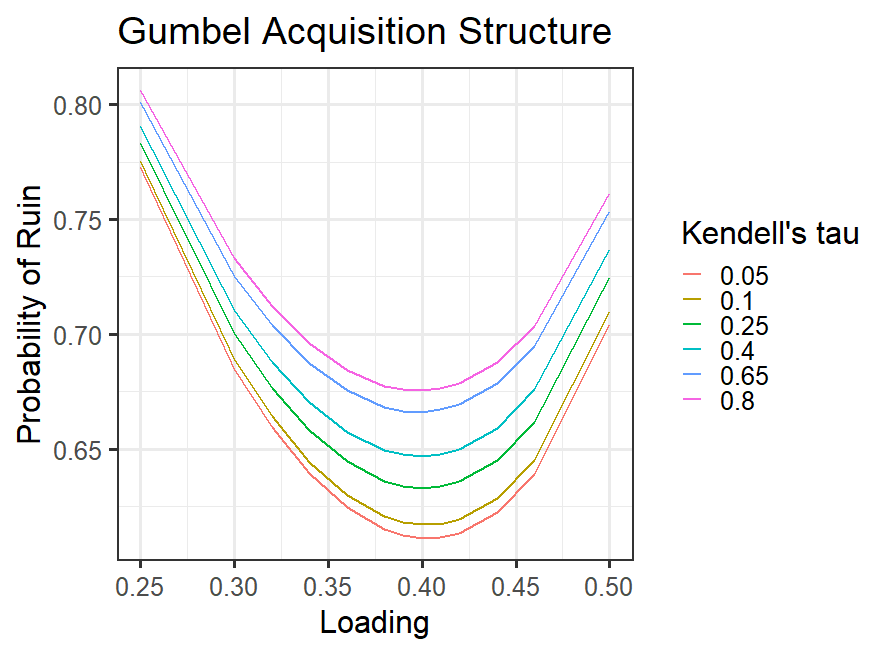}  
\end{subfigure}
  \caption[Acquisition dependency.]{The ruin probability when the acquisition is modelled with a Clayton (left) and a Gumbel (right) dependency structure. The surplus is constant. In both cases, the ruin probability is higher for higher kendell's tau. The Gumbel case is more riskier. The surplus is fixed at 5000 units.}
  \label{fig:acquisition}
\end{figure}

In Figure \ref{fig:acquisition} we can see the ruin probability for different dependency values when the surplus is fixed at 5000 units. We can see that the ruin probability is higher for more dependent acquisition, as we expected. Also, we can see that the Gumbel acquisition model gives higher ruin probabilities than the Clayton model for the same Kendell's tau value. They give the same optimal loading parameter. It can also be seen that when kendell's tau is $0.05$ (close to $0$) the ruin probability is close to the case of independent acquisition, as expected. The optimal loading parameter is the same for all dependency levels.


\appendix

\section{Numerical scheme for equation \ref{eq:Ruin_lundberg_inf_survival}, using linear approximation}
\label{appendix:NA_ez_Poi_Exp_inf}
Consider the process 

$$ X_t = u + ct - \sum_{i=0}^{N_t} Y_i$$ 

where $Y_i$ are iid continuous random variables and $N_t$ is a $Poisson(\lambda t)$. To approximate equation \ref{eq:Ruin_lundberg_inf_survival}, take a grid of points  $\epsilon=x_0<x_1<...x_n$, $x_i \in \mathbbm{R}, \forall i \in \mathbbm{N}, \epsilon >0$, with equal interval lengths, $h = x_{i}- x_{i-1}$. A linear approximation is used to approximate $\Bar{V}(x)$

\[
\Bar{V}(z) =
 \Bar{V}(x_{j-1}) + \frac{\Bar{V}(x_j)- \Bar{V}(x_{j-1})}{h}(z-x_{j-1}), \quad z \in [x_{j-1},x_j], j\leq i 
\]

where $ \frac{\Bar{V}(x_j)- \Bar{V}(x_{j-1})}{h}$ is an approximation of the derivative, $\Bar{V}'(x_{j-1})$, using the so-called forward difference

Let $\Bar{V}_i$ denote the approximation of $\Bar{V}(x_i)$. Let $\Bar{S}(x) = \int_0^x \Bar{F}(y)dy$ and $\Bar{\Bar{S}}(x) = \int_0^x \Bar{S}(y)dy$. For each $x_i, i>0$ solve the following equation

\begin{equation*}
        \Bar{V}(x_i)-\Bar{V}(0^+) = \frac{\lambda}{c}\int_0^x\Bar{V}(x_i-y)\Bar{F}(y)dy.
\end{equation*} 

The goal is to develop a recursive method from $x_0$ as the value of $\Bar{V}_0$ is known. \\

$\boldsymbol{if \quad i =0}$ \\

Set $\Bar{V}_0 = 1-\frac{\lambda}{c}\E[Y]$ \\

$\boldsymbol{if \quad i =1}$ \\

Calculate
\begin{equation*}
    \begin{split}
    \Bar{V}_1 &= \Bar{V}_0 + \frac{\lambda}{c}\int_{o}^{x_i} ((\Bar{V}_0 + \frac{\Bar{V}_1- \Bar{V}_0}{h})(x_1-y-x_0))\Bar{F}(y)dy \\
    &= \Bar{V}_0 + \frac{\lambda}{c} \Big(\Bar{V}_0(\Bar{S}(x_1)-\Bar{S}(x_0)) + \\
    & \qquad \frac{\Bar{V}_1-\Bar{V}_0}{h}\big([x_1-y-x_0]_{x_0}^{x_1}+ \Bar{\Bar{S}}(x_1)-\Bar{\Bar{S}}(x_0) \big) \Big)\\
    &= \Bar{V}_0 + \frac{\lambda}{c}( a_{1,1} + \frac{\Bar{V}_1-\Bar{V}_0}{h}a_{2,1})
    \end{split}
\end{equation*}
$$\Leftrightarrow$$
$$(1-\frac{\lambda a_{2,1}}{ch})\Bar{V}_1 =\Bar{V}_0 + \frac{\lambda}{c}\Bar{V}_0(a_{1,1}-\frac{a_{2,1}}{h})  $$

$\boldsymbol{if \quad i >1}$ \\

Calculate

\begin{equation*}
    \begin{split}
    \Bar{V}(x_i)-\Bar{V}(0) &= \frac{\lambda}{c}\Big(\sum_{j=2}^i\int_{x_{j-1}}^{x_j}\Bar{V}(x_i-y)S(y)dy + \int_{x_{i-1}}^{x_i}\Bar{V}(x_i-y)S(y)dy \Big) \\
    &= \frac{\lambda}{c}\bigg( \sum_{j=2}^i\Big(\Bar{V}_{i-j}(\Bar{S}(x_j)-\Bar{S}(x_{j-1})) + \frac{\Bar{V}_{i-j+1}-\Bar{V}_{i-j}}{h}\big([x_i-y-x_{i-j}]_{x_{j-1}}^{x_j}+ \Bar{\Bar{S}}(x_j)-\Bar{\Bar{S}}(x_{j-1}) \big) \Big) +\\
    & \qquad \Bar{V}_{i-j}(\Bar{S}(x_i)-\Bar{S}(x_{i-1})) + \frac{\Bar{V}_i-\Bar{V}_{i-j}}{h}\big([x_i-y-x_{i-j}]_{x_{i-1}}^{x_i}+ \Bar{\Bar{S}}(x_i)-\Bar{\Bar{S}}(x_{i-1}) \big)\bigg) \\
    &= \frac{\lambda}{c}\bigg( \sum_{j=2}^i\Big(\Bar{V}_{i-j}(\Bar{S}(x_j)-\Bar{S}(x_{j-1})) + \frac{\Bar{V}_{i-j+1}-\Bar{V}_{i-j}}{h}\big([x_i-y-x_{i-j}]_{x_{j-1}}^{x_j}+ \Bar{\Bar{S}}(x_j)-\Bar{\Bar{S}}(x_{j-1}) \big) \Big) +\\
    & \qquad \Bar{V}_{i-j}a_{1,i} + \frac{\Bar{V}_i-\Bar{V}_{i-j}}{h}a_{2,i}\bigg)
    \end{split}
\end{equation*}
$$\Leftrightarrow$$
\begin{equation*}
    \begin{split}
    (1-\frac{\lambda a_{2,i}}{ch})\Bar{V}_i & =\Bar{V}_0 + \\
    &  \qquad \frac{\lambda}{c}\bigg( \sum_{j=2}^i\Big(\Bar{V}_{i-j}(\Bar{S}(x_j)-\Bar{S}(x_{j-1})) + \frac{\Bar{V}_{i-j+1}-\Bar{V}_{i-j}}{h}\big([x_i-y-x_{i-j}]_{x_{j-1}}^{x_j}+ \Bar{\Bar{S}}(x_j)-\Bar{\Bar{S}}(x_{j-1}) \big) \Big) + \\
    & \qquad \Bar{V}_{i-j}a_{1,i} - \frac{\Bar{V}_{i-j}}{h}a_{2,i}\bigg)
    \end{split}
\end{equation*}

\begin{algorithm}[H]
\label{algo:infnite_survival}
\SetAlgoLined
Let the symbol $\Vec{}$ denote a vector. \\
\vspace{2mm}
\textbf{Initialize} \\
$\Vec{x}$ for some $x_0,...,x_n$ \\
$\Vec{\Bar{V}}$ with length equal to the length of $\Vec{x}$ \\
$N_x \gets$ length of $\Vec{x}$ \\
\vspace{2mm}
\textit{\# loop to estimate each value in $\Vec{\Bar{V}}$ } \\
\For{$i$ in $0,...,(N_x-1)$}{
  \uIf{i = 0}{
    $\Bar{V}_0 \gets 1-\frac{\lambda}{c}\E[Y]$ 
  }
  \uElseIf{i = 1}{
    $\Bar{V}_1 \gets$ take case $i = 1$ from above and isolate $\Bar{V}_1$\\
  }
  \Else{
    $\Bar{V}_i \gets$ take case $i > 1$ from above and isolate $\Bar{V}_i$ \\
  }
}

\textbf{Return} $\Vec{\Bar{V}}$, $\Vec{x}$ 

 \caption{\textbf{Estimation of $\Bar{V}(x)$}}
\end{algorithm}

\section*{Acknowledgements}
The second and third authors acknowledge financial support from FCT – Funda\c c\~ao para a Ci\^encia e Tecnologia (Portugal), national funding, through research grant UIDB/05069/2020.


\end{document}